\documentclass[onecolumn,12pt]{IEEEtran}
\usepackage{amssymb}
\usepackage{amsmath,amsthm}
\usepackage{stfloats}
\usepackage{graphicx}
\usepackage{subfigure}
\usepackage{tabularx}
\usepackage{epsfig,epsf,color,balance,cite}
\usepackage{multirow}
\usepackage{bm}

\newtheorem{remark}{Remark}
\newtheorem{theorem}{Theorem}

\newtheorem{corollary}{Corollary}

\begin{document}
%
\title{Secrecy Capacity Bounds for Visible Light Communications With Signal-Dependent Noise}

\author{Jin-Yuan Wang,  \IEEEmembership{Member, IEEE},
        Xian-Tao Fu,
        Jun-Bo Wang, \IEEEmembership{Member, IEEE},\\
        Min Lin,  \IEEEmembership{Member, IEEE},
        Julian Cheng, \IEEEmembership{Senior Member, IEEE},
        and Mohamed-Slim Alouini, \IEEEmembership{Fellow, IEEE}

\thanks{Jin-Yuan Wang is with Key Laboratory of Broadband Wireless Communication and Sensor Network Technology, Nanjing University of Posts and Telecommunications, Nanjing 210003, China, and also with Shanghai Key Laboratory of Trustworthy Computing, East China Normal University, Shanghai 200062, China (corresponding author, e-mail: jywang@njupt.edu.cn).}
\thanks{Xian-Tao Fu and Min Lin are with College of Telecommunications \& Information Engineering, Nanjing University of Posts and Telecommunications, Nanjing 210003, China (e-mail: \{1019010206, linmin\}@njupt.edu.cn).}
\thanks{Jun-Bo Wang is with National Mobile Communications Research Laboratory, Southeast University, Nanjing 211111, China (e-mail: jbwang@seu.edu.cn).}
\thanks{Julian Cheng is with School of Engineering, The University of British Columbia, Kelowna, BC, V1V 1V7, Canada (e-mail: julian.cheng@ubc.ca).}
\thanks{Mohamed-Slim Alouini is with Computer, Electrical and Mathematical Science and Engineering Division, King Abdullah University of Science and Technology, Thuwal 23955-6900, Saudi Arabia (e-mail: slim.alouini@kaust.edu.sa).}
}


\maketitle \linespread{1.45}
\begin{abstract}
In physical-layer security, one of the most fundamental issues is the secrecy capacity.
The objective of this paper is to determine the secrecy capacity for an indoor visible light communication system consisting of a transmitter, a legitimate receiver and an eavesdropping receiver.
In such a system, both signal-independent and signal-dependent Gaussian noises are considered.
Under non-negativity and average optical intensity constraints, lower and upper secrecy capacity bounds are first derived by the variational method, the dual expression of the secrecy capacity, and the concept of ``the optimal input distribution that escapes to infinity''.
Numerical results show that the secrecy capacity upper and lower bounds are tight.
By an asymptotic analysis at large optical intensity, there is a small performance gap between the asymptotic upper and lower bounds.
Then, by adding a peak optical intensity constraint, we further analyze the exact and asymptotic secrecy capacity bounds. Finally, the tightness of the derived bounds is verified by numerical results.
\end{abstract}


\IEEEpeerreviewmaketitle

\newpage
\baselineskip=8.5mm

\section{Introduction}
\label{section1}
With the rapid development of the fifth generation wireless network and the application of Internet of things technology,
new transmission technologies have been developed to meet the demands of explosive growth in data traffic.
One such new transmission technology is indoor visible light communication (VLC), which uses visible light generated by light-emitting diodes (LEDs) for communication and is considered as a promising solution for the pressing data traffic demands \cite{BIB01,BIB02}.

Different from conventional radio frequency wireless communications (RFWC), indoor VLC has several distinctive properties.
First, the optical intensity of the input signal in VLC is controlled to convey wireless information.
Second, the indoor VLC performs communication and illumination simultaneously.
The average optical intensity of indoor VLC should not fluctuate with time to satisfy the indoor illumination requirement.
Third, since the optical intensity cannot be negative, the input signal in indoor VLC should be non-negative.
Consequently, the developed theory and analysis in traditional RFWC are not directly applicable to indoor VLC.

For conventional RFWC systems, the performance is usually evaluated by the classic concept of ``Shannon capacity" \cite{BIB02_1}, a concept that predicts the maximum transmission rate via a channel at a given noise level.
Unfortunately, such a concept is unsuitable for VLC systems due to the distinctive properties of indoor VLC.
What is the channel capacity limit of indoor VLC?
Recently, the channel capacity of VLC was analyzed by the inverse source coding approach \cite{BIB03}.
However, the derived channel capacity is not in closed-form, and the evaluation of channel capacity is time-consuming.
While tractable and closed-form expressions of channel capacity bounds were derived \cite{BIB04,BIB05}, the LED's peak optical intensity constraint, which reflects the LED's maximum luminous ability, was not considered.
With additional peak constraint,
tight channel capacity bounds can also be derived along with the corresponding capacity-achieving input distribution \cite{BIB06,BIB12}.
However, these bounds do not consider the effects of the modulation schemes on the capacity performance.
Employing pulse amplitude modulation \cite{BIB07}, orthogonal frequency division multiplexing \cite{BIB08}, and color shift keying \cite{BIB09}, the authors investigated the channel capacities of VLC, respectively.
In addition, by considering some actual factors, such as spatially random receiver \cite{BIB10} and signal interference \cite{BIB11}, the authors further evaluated the channel capacities of VLC.
All these works \cite{BIB03,BIB04,BIB05,BIB06,BIB12,BIB07,BIB08,BIB09,BIB10,BIB11} assumed noises are independent of the input signal.
This assumption is reasonable if the ambient light is dominant or the receiver suffers from intense thermal noise.
However, in practical VLC systems, common indoor environments desire high received optical intensity to satisfy the illumination requirement \cite{BIB13,BIB14}.
At such a high optical intensity, this assumption neglects a basic issue:
the noise strength relies on the input signal due to the random nature of photon emission of the LEDs \cite{BIB15}.
By considering the signal-dependent noise, we further analyzed the channel capacity bounds of indoor VLC \cite{BIB16}.
In short, all these channel capacity bounds \cite{BIB03,BIB04,BIB05,BIB06,BIB12,BIB07,BIB08,BIB09,BIB10,BIB11,BIB13,BIB14,BIB15,BIB16} provide theoretical references for designing practical VLC systems having high data rates \cite{BIB17,BIB18,BIB19}.

Compared with conventional RFWC systems, VLC systems not only have higher data rate,
but also can provide better security.
This is because light cannot penetrate through the walls.
Despite better signal confinement,
the VLC channels still have open and broadcast features \cite{BIB20}.
As a result,
information security issue in VLC have become a major concern of network administrators.
Conventional security schemes are generally performed at the upper-layers of the network stack via password protection, access controls, and end-to-end encryption.
The safety of these schemes depends on the restricted storage capacity and computational power of eavesdroppers.
However, the encrypted data may arouse suspicion, and even the most theoretically robust encryption can be defeated by eavesdroppers using non-computational approaches such as side-channel analysis.
Recently, as opposed to traditional network security, physical-layer security (PLS) has been proposed.
The basic principle of PLS is to exploit the wireless channel characteristics to ensure the successful decoding of the legitimate receiver and prevent the eavesdropper from doing so.
Similar to Shannon capacity, the PLS theory was also first developed by Shannon \cite{BIB21}.
After that, the widely used concept ``secrecy capacity" was proposed by Wyner \cite{BIB22},
and then the research was extended to many RFWC scenarios \cite{BIB23,BIB24,BIB25,BIB26,BIB27}.
Although much work has been done, the derived PLS results in RFWC cannot be directly applied to VLC.

To determine the PLS performance of VLC,
we derived the upper and lower bounds on secrecy capacity \cite{BIB28}.
However, the channel model considers only the specular reflection.
With consideration of both specular and diffusive reflections of VLC channels,
a modified Monte-Carlo ray tracing channel model was recently proposed to derive secrecy capacity bounds of VLC channels \cite{BIB29}.
The secrecy capacity analysis of VLC was also extended to the multiple-input single-output scenario \cite{BIB30} and
the spatially random transceiver scenario \cite{BIB31}.
Moreover, the PLS of VLC was comprehensively discussed \cite{BIB31_1}.
In these works \cite{BIB28,BIB29,BIB30,BIB31,BIB31_1}, the corrupting noises are assumed to be independent of the signal, but the signal-dependent noise is ignored.
For an VLC system having the signal-dependent noise, the secrecy-capacity-achieving input distribution and the asymptotic secrecy capacities were further discussed \cite{BIB32,BIB32_1}.
However, the derived results depend on the assumption that the signal-dependent noises of the main channel and the eavesdropping channel are identical. Moreover, exact closed-form expressions of the secrecy capacity bounds have not yet been derived.

In this paper, we further analyze the secrecy capacity of a classic three-node indoor VLC system having the signal-dependent noise.
Without assuming identical signal-dependent noise variances at the legitimate receiver and the eavesdropper \cite{BIB32,BIB32_1}, we consider arbitrary noise variances. The main contributions of this paper are summarized as follows:
\begin{enumerate}
  \item We analyze the secrecy capacity bounds for the VLC system by considering the non-negativity and average optical intensity constraints. By using the variational method, we first derive a lower bound on secrecy capacity. Applying the dual expression of the secrecy capacity and the concept of ``the optimal distribution that escapes to infinity'', we then provide an upper bound on the secrecy capacity. Numerical results verify the accuracy of the derived theoretical bounds.
  \item By adding an additional peak optical intensity constraint, we further investigate the secrecy capacity bounds for the VLC. Based on the information theory, we derive novel lower and upper bounds on secrecy capacity. The accuracy of the derived theoretical expressions is also confirmed by numerical results.
  \item We analyze the asymptotic behaviors at high optical intensity. Based on the exact secrecy capacity bounds, we derive asymptotic secrecy capacity bounds when the optical intensity tends to infinity. Theoretical analysis shows that a small performance gap exists between the asymptotic upper and lower bounds on secrecy capacity.
\end{enumerate}

The reminder of this paper is organized as follows.
Section \ref{section2} details the system model.
Section \ref{section3} analyzes the secrecy capacity bounds when considering the non-negativity and average optical intensity constraints,
while Section \ref{section4} derives the secrecy capacity bounds when considering both non-negativity, average optical intensity and peak optical intensity constraints.
Numerical results are provided in Section \ref{section5}.
Finally, conclusions of the paper are presented in Section \ref{section6}.

\emph{Notations:} In this paper, ${\cal N}( {\mu ,{\sigma ^2}})$ denotes a Gaussian distribution having mean $\mu $ and variance ${\sigma ^2}$;
${f_X}(\cdot)$ and ${f_{Y\left| X \right.}}(\cdot)$ denote the probability density function (PDF) of $X$ and the conditional PDF of $Y$ given variable $X$;
${E_X}\left(  \cdot  \right)$ denotes the expectation operator with respect to $X$;
${\mathop{\rm var}} \left(  \cdot  \right)$ denotes the variance of a variable;
$I\left( { \cdot ; \cdot } \right)$ denotes the mutual information;
${\cal H}\left(  \cdot  \right)$ and ${\cal H}\left( { \cdot \left|  \cdot  \right.} \right)$ denote the entropy and the conditional entropy;
$D( \cdot || \cdot )$ denotes the relative entropy.
We use $\ln \left(  \cdot  \right)$ for the natural logarithm and $Ei(x) = \int_{ - \infty }^x {e^t}/t {\rm d}t$ for the exponential integral function [36, P883].

\section{System Model}
\label{section2}
Consider an indoor VLC system consisting of a transmitter (i.e., Alice), a legitimate receiver (i.e., Bob) and an eavesdropping receiver (i.e., Eve), as shown in Fig. \ref{fig1}.
In such a system, Alice is equipped with a single LED to transmit optical intensity signals to Bob in the presence of Eve.
Both Bob and Eve are equipped with one photodiode (PD) individually to perform the optical-to-electrical conversion.
At the receiver of Bob or Eve, the main noise includes signal-independent and signal-dependent noises \cite{BIB28}.
The received signals at Bob and Eve can be expressed as
\begin{eqnarray}
\left\{ {\begin{array}{*{20}{c}}
{{Y_{\rm{B}}} = {H_{\rm{B}}}X + \sqrt {{H_{\rm{B}}}X} {Z_{{\rm{B,1}}}} + {Z_{{\rm{B,0}}}}}\\
{{Y_{\rm{E}}} = {H_{\rm{E}}}X + \sqrt {{H_{\rm{E}}}X} {Z_{{\rm{E,1}}}} + {Z_{{\rm{E,0}}}}}
\end{array}} \right.,
\label{eq1}
\end{eqnarray}
where $X$ denotes the transmit optical intensity signal from Alice.
${H_{\rm{B}}}$ and ${H_{\rm{E}}}$ denote the channel gains of the main channel (i.e., Alice-Bob channel) and the eavesdropping channel (i.e., Alice-Eve channel), respectively.
At Bob, ${Z_{{\rm{B,0}}}} \sim {\cal N}\left( {0,\sigma _{\rm{B}}^2} \right)$ and ${Z_{{\rm{B,1}}}} \sim {\cal N}\left( {0,\varsigma _{\rm{B}}^2\sigma _{\rm{B}}^2} \right)$ stand respectively for the signal-independent and signal-dependent Gaussian noises, where $\varsigma _{\rm{B}}^2 > 0$ denotes the ratio of the signal-dependent noise variance to the signal-independent noise variance at Bob.
Similarly, ${Z_{{\rm{E,0}}}} \sim {\cal N}\left( {0,\sigma _{\rm{E}}^2} \right)$ and ${Z_{{\rm{E,1}}}} \sim {\cal N}\left( {0,\varsigma _{\rm{E}}^2\sigma _{\rm{E}}^2} \right)$ denote the signal-independent and signal-dependent Gaussian noises at Eve, where $\varsigma _{\rm{E}}^2 > 0$ is the ratio of the signal-dependent noise variance to the signal-independent noise variance at Eve.
Furthermore, we also assume ${Z_{{\rm{B,0}}}}$ and ${Z_{{\rm{E,0}}}}$ to be independent of each other.

\begin{figure}
\centering
\rotatebox[origin=c]{-90}{\includegraphics[width=8cm]{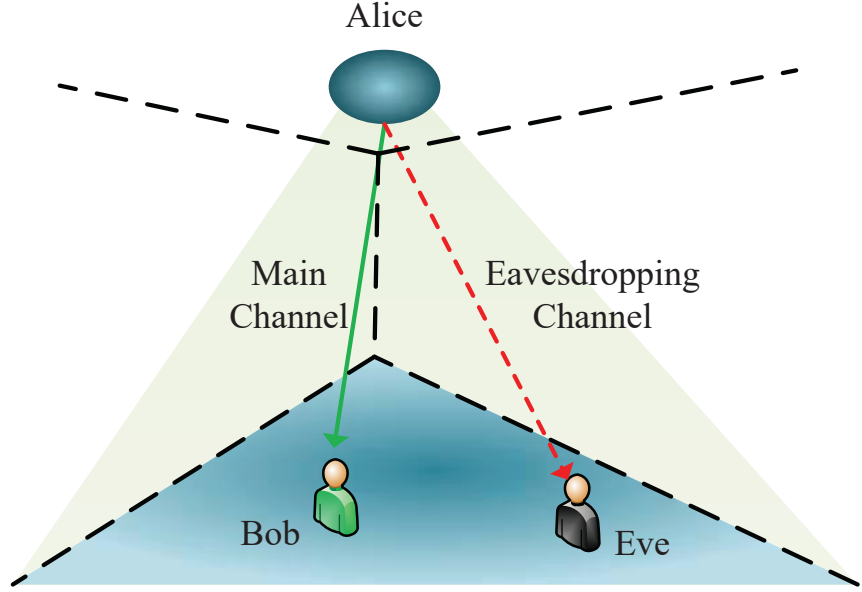}}
\caption{An indoor VLC system with Alice, Bob and Eve}
\label{fig1}
\end{figure}

For indoor VLC, we consider the following three signal constraints \cite{BIB28}.
\begin{description}
  \item[1)] \emph{Non-negativity}: In the received signal model (\ref{eq1}), the input signal $X$ is a non-negative random variable representing the intensity of the optical signal. Therefore, the non-negativity constraint is given by
\begin{eqnarray}
X \ge 0.
\label{eq2}
\end{eqnarray}
  \item[2)] \emph{Peak optical intensity constraint}: Because of the practical and safety restrictions, the intensity of the input signal is generality constrained by a peak optical intensity constraint. Therefore, we have
\begin{eqnarray}
X \le A,
\label{eq3}
\end{eqnarray}
where $A$ is the peak optical intensity of the LED at Alice.
  \item[3)] \emph{Average optical intensity constraint}: Due to the illumination requirement in VLC, the average optical intensity cannot fluctuate with time but can be adjusted according to the users' requirement. Therefore, the average optical intensity constraint is expressed as
\begin{eqnarray}
 {E_X}{\rm{(}}X{\rm{) = }}\xi P, \label{eq4}
\end{eqnarray}
where $\xi  \in \left( {0,1} \right]$ is the dimming target, $P \in \left( {0,A} \right]$ is the nominal optical intensity of the LED at Alice.
\end{description}

In VLC, the received optical intensity of the line-of-sight (LoS) path dominates that of the reflection paths, and thus we consider only the LoS path and ignore reflections from surrounding surfaces \cite{BIB35_1}.
In the received signal model (\ref{eq1}), the LoS channel gain ${H_k}$ ($k = {\rm{B}}$ or ${\rm{E}}$) can be expressed as
\begin{eqnarray}
{H_k} = \left\{ \begin{array}{l}
\frac{{(m + 1){A_r}}}{{2\pi D_k^2}}{T_s}g{\cos ^m}({\varphi _k})\cos ({\psi _k}),\;{\rm{if}}\;0 \le {\psi _k} \le \Psi \\
\;\;\;\;\;\;\;\;\;\;\;\;\;\;\;\;\;\;\;\;\;0,\;\;\;\;\;\;\;\;\;\;\;\;\;\;\;\;\;\;\;\;\;{\rm{otherwise}}
\end{array} \right.,
\label{eq5}
\end{eqnarray}
where $m$ is the order of Lambertian emission;
${A_r}$ is the physical area of the PD;
${T_s}$ and $g$ are the optical filter gain and concentrator gain of the PD, respectively;
$\Psi $ is the field of view of the PD;
${D_k}$, ${\varphi _k}$ and ${\psi _k}$ are the link distance, the irradiance angle and the incidence angle from Alice to Bob ($k = {\rm{B}}$) or Eve ($k = {\rm{E}}$).

\section{Secrecy Capacity for VLC Having Non-negativity and Average Optical Intensity Constraints}
\label{section3}
In this section, we analyze the secrecy capacity of the considered VLC system having the non-negativity in (\ref{eq2}) and the average optical intensity constraint in (\ref{eq4}). Specifically, we will provide tight secrecy capacity bounds and investigate the asymptotic behavior in high optical intensity regime.

The secrecy capacity represents the maximum transmission rate at which the legitimate receiver can reliably decode the transmitted message,
while the eavesdropping receiver cannot infer information at any positive rate \cite{BIB33}.
If the main channel is worse than the eavesdropping channel, then the random variables $X$, ${Y_{\rm{B}}}$ and ${Y_{\rm{E}}}$ form a Markov chain $X \to {Y_{\rm{E}}} \to {Y_{\rm{B}}}$, and thus the secrecy capacity is zero; otherwise, the random variables $X$, ${Y_{\rm{B}}}$ and ${Y_{\rm{E}}}$ form a Markov chain $X \to {Y_{\rm{B}}} \to {Y_{\rm{E}}}$,
and thus the main channel is stochastically degraded with respect to the eavesdropping channel.
In this case, we can derive the secrecy capacity by solving the following functional optimization problem
\begin{eqnarray}
&&{C_{\rm s}} = \mathop {\max }\limits_{{f_X}\left( x \right)} \left[ {I\left( {X;{Y_{\rm{B}}}} \right) - I\left( {X;{Y_{\rm{E}}}} \right)} \right] \nonumber\\
{\rm{s}}{\rm{.t}}{\rm{.}}\;&&\int_0^\infty  {{f_X}\left( x \right){\rm{d}}x}  = 1 \label{eq6}\\
&&{E_X}\left( X \right) = \int_0^\infty  {x{f_X}\left( x \right){\rm{d}}x}  = \xi P, \nonumber
\end{eqnarray}
where ${C_{\rm s}}$ denotes the secrecy capacity, and ${f_X}\left( x \right)$ denotes the PDF of $X$.
From \cite{BIB32}, there exists a unique discrete input PDF $f_X^*\left( x \right)$ that achieves the secrecy capacity of problem (\ref{eq6}).
However, the exact expression of $f_X^*\left( x \right)$ is still unknown.
Consequently, it is challenging to obtain a closed-form expression of the secrecy capacity for problem (\ref{eq6}).
Alternatively, we will derive lower and upper bounds on the secrecy capacity.

\subsection{Lower Bound on Secrecy Capacity}
\label{section3_1}
In the optimization problem (\ref{eq6}), by choosing an arbitrary input PDF $f_X(x)$ satisfying constraints (\ref{eq2}) and (\ref{eq4}), we can obtain a lower bound on the secrecy capacity as
\begin{eqnarray}
{C_{\rm s}} &\ge& {\left. {\left[ {I\left( {X\;;{Y_{\rm{B}}}} \right) - I\left( {X\;;{Y_{\rm{E}}}} \right)} \right]} \right|_{\forall {f_X}\left( x \right)\;{\rm{satifies}}\;(2)\;{\rm and}\;(4)}} \nonumber\\
 &=& {\cal H}\left( {{Y_{\rm{B}}}} \right) - {\cal H}\left( {{Y_{\rm{B}}}\left| X \right.} \right) - {\cal H}\left( {{Y_{\rm{E}}}} \right){\rm{ + }}{\cal H}\left( {{Y_{\rm{E}}}\left| X \right.} \right).
\label{eq7}
\end{eqnarray}

According to the received signal model (\ref{eq1}), the conditional PDF ${f_{{Y_k}\left| X \right.}}\left( {y\left| X \right.} \right)$ is given by
\begin{equation}
{f_{{Y_k}\left| X \right.}}\left( {y\left| X \right.} \right) = \frac{{\exp \left[ { - \frac{{{{\left( {y - {H_k}X} \right)}^2}}}{{2\left( {1 + {H_k}X\varsigma _k^2} \right)\sigma _k^2}}} \right]}}{{\sqrt {2\pi \left( {1 + {H_k}X\varsigma _k^2} \right)\sigma _k^2} }},\;k = {\rm{B}}\;{\rm{or}}\;{\rm{E}}{\rm{.}}
\label{eq8}
\end{equation}
Then, the conditional entropy ${\cal H}\left( {{Y_k}\left| X \right.} \right)$ in (\ref{eq7}) is written as
\begin{equation}
{\cal H}\left( {{Y_k}\left| X \right.} \right){\rm{ = }}\frac{1}{2}\ln \left( {2\pi e\sigma _k^2} \right){\rm{ + }}\frac{1}{2}{E_X}\left[ {\ln \left( {1 + {H_k}\varsigma _k^2X} \right)} \right],\;k = {\rm{B}}\;{\rm{or}}\;{\rm{E}}.
\label{eq9}
\end{equation}

Furthermore, ${\cal H}\left( {{Y_{\rm{E}}}} \right)$ is upper-bounded by the differential entropy of a Gaussian random variable having a variance ${\mathop{\rm var}} \left( {{Y_{\rm{E}}}} \right)$ [38, Theorem 8.6.5], i.e.,
\begin{equation}
{\cal H}\left( {{Y_{\rm{E}}}} \right) \le \frac{1}{2}\ln \left[ {2\pi e{\mathop{\rm var}} \left( {{Y_{\rm{E}}}} \right)} \right].
\label{eq10}
\end{equation}

According to (9) in \cite{BIB34}, we know that the output entropy is always larger than the input entropy (OELIE).
As a result, ${\cal H}\left( {{Y_{\rm{B}}}} \right)$ can be lower-bounded by
\begin{equation}
{\cal H}\left( {{Y_{\rm{B}}}} \right) \ge {\cal H}\left( X \right) + {f_{{\rm{low}}}}\left( {{H_{\rm{B}}},\xi ,P} \right),
\label{eq11}
\end{equation}
where ${f_{{\rm{low}}}}\left( {{H_{\rm{B}}},\xi ,P} \right)$ is given by
\begin{eqnarray}
{f_{{\rm{low}}}}\left( {{H_{\rm{B}}},\xi ,P} \right) &=& \frac{1}{2}\ln \left( {H_{\rm{B}}^2 + \frac{{2{H_{\rm{B}}}\varsigma _{\rm{B}}^2\sigma _{\rm{B}}^2}}{{\xi P}}} \right) - \frac{{{H_{\rm{B}}}\xi P + \varsigma _{\rm{B}}^2\sigma _{\rm{B}}^2}}{{\varsigma _{\rm{B}}^2\sigma _{\rm{B}}^2}} \nonumber\\
 &&+ \frac{{\sqrt {{H_{\rm{B}}}\xi P\left( {{H_{\rm{B}}}\xi P + 2\varsigma _{\rm{B}}^2\sigma _{\rm{B}}^2} \right)} }}{{\varsigma _{\rm{B}}^2\sigma _{\rm{B}}^2}}.
\label{eq12}
\end{eqnarray}

Substituting (\ref{eq9}), (\ref{eq10}) and (\ref{eq11}) into (\ref{eq7}), we can rewrite the lower bound on secrecy capacity as
\begin{eqnarray}
{C_s} &\ge& {\cal H}\left( X \right) + {f_{{\rm{low}}}}\left( {{H_{\rm{B}}},\xi ,P} \right) + \frac{1}{2}\ln \left( {\frac{{\sigma _{\rm{E}}^2}}{{\sigma _{\rm{B}}^2}}} \right) \nonumber\\
 &&+ \frac{1}{2}{E_X}\left[ {\ln \left( {\frac{{1 + {H_{\rm{E}}}\varsigma _{\rm{E}}^2X}}{{1 + {H_{\rm{B}}}\varsigma _{\rm{B}}^2X}}} \right)} \right] - \frac{1}{2}\ln \left[ {2\pi e{\mathop{\rm var}} \left( {{Y_{\rm{E}}}} \right)} \right].
\label{eq13}
\end{eqnarray}
As can be seen from (\ref{eq13}), we can obtain a lower bound on secrecy capacity by choosing an arbitrary input PDF satisfying the constraints in problem (\ref{eq6}).
However, we should choose a good input PDF to obtain a tight lower bound.
By using the variational method, we derive a tight secrecy capacity lower bound in the following theorem.

\begin{theorem}
For indoor VLC having constraints (\ref{eq2}) and (\ref{eq4}), we derive a lower bound on the secrecy capacity as
\begin{eqnarray}
{C_{{\rm{Low}}}} &=& \frac{1}{2}\ln \left[ {\frac{{e{\xi ^2}{P^2}\sigma _{\rm{E}}^2}}{{2\pi \sigma _{\rm{B}}^2\left( {H_{\rm{E}}^2{\xi ^2}{P^2} + {H_{\rm{E}}}\xi P\varsigma _{\rm{E}}^2\sigma _{\rm{E}}^2 + \sigma _{\rm{E}}^2} \right)}}} \right] + {f_{{\rm{low}}}}\left( {{H_{\rm{B}}},\xi ,P} \right) \nonumber\\
 &&+ \frac{1}{2}\left[ {{e^{\frac{1}{{{H_{\rm{B}}}\varsigma _{\rm{B}}^2\xi P}}}}Ei\left( { - \frac{1}{{{H_{\rm{B}}}\varsigma _{\rm{B}}^2\xi P}}} \right) - {e^{\frac{1}{{{H_{\rm{E}}}\varsigma _{\rm{E}}^2\xi P}}}}Ei\left( { - \frac{1}{{{H_{\rm{E}}}\varsigma _{\rm{E}}^2\xi P}}} \right)} \right].
\label{eq14}
\end{eqnarray}
\label{the1}
\end{theorem}

\begin{proof}
See Appendix \ref{Appendix_A}.
\end{proof}

\begin{corollary}
In \emph{Theorem \ref{the1}}, when we ignore the signal-dependent noise (i.e., ${\varsigma _{\rm{B}}} \to 0$ and ${\varsigma _{\rm{E}}} \to 0$), the secrecy capacity lower bound in (\ref{eq14}) reduces to
\begin{equation}
\mathop {\lim }\limits_{\scriptstyle{\varsigma _{\rm{B}}} \to 0\hfill\atop
\scriptstyle{\varsigma _{\rm{E}}} \to 0\hfill} {C_{{\rm{Low}}}} = \frac{1}{2}\ln \left[ {\frac{{eH_{\rm{B}}^2{\xi ^2}{P^2}\sigma _{\rm{E}}^2}}{{2\pi \sigma _{\rm{B}}^2\left( {H_{\rm{E}}^2{\xi ^2}{P^2} + \sigma _{\rm{E}}^2} \right)}}} \right].
\label{eq15}
\end{equation}
\label{cor1}
\end{corollary}

\begin{proof}
The proof of \emph{Corollary \ref{cor1}} is straightforward and hence omitted.
\end{proof}

\begin{remark}
For the VLC having only the signal-independent noise,
we derive the secrecy capacity lower bound (\ref{eq15}) in \emph{Corollary \ref{cor1}} based on the principle of OELIE (\ref{eq11}),
while the authors in \cite{BIB28} also derived a lower bound (8) by using the entropy power inequality (EPI).
It can be easily shown that the lower bound (\ref{eq15}) in \emph{Corollary \ref{cor1}} is smaller than the lower bound (8) in \cite{BIB28}.
This finding suggests that the EPI is a more efficient approach to analyze the secrecy capacity lower bound when considering only the signal-independent noise.
Unfortunately, the EPI approach cannot be applied to the PLS analysis of VLC having the signal-dependent noise.
\label{rem1}
\end{remark}

\subsection{Upper Bound on Secrecy Capacity}
\label{section3_2}
In this subsection, the derivation of the upper bound is based on the dual expression of the secrecy capacity.
For an arbitrary conditional PDF ${g_{{Y_{\rm{B}}}|{Y_{\rm E}}}}({y_{\rm{B}}}|{Y_{\rm{E}}})$, the following inequality holds \cite{BIB28}
\begin{equation}
I(X;{Y_{\rm{B}}}|{Y_{\rm{E}}}) \le {E_{X{Y_{\rm{E}}}}}\left[ {D\left( {\left. {{f_{{Y_{\rm{B}}}|X{Y_{\rm E}}}}({y_{\rm{B}}}|X,{Y_{\rm{E}}})} \right\|{g_{{Y_{\rm{B}}}|{Y_{\rm E}}}}({y_{\rm{B}}}|{Y_{\rm{E}}})} \right)} \right].
\label{eq16}
\end{equation}
In the inequality (\ref{eq16}), selecting any ${g_{{Y_{\rm{B}}}|{Y_{\rm E}}}}({y_{\rm{B}}}|{Y_{\rm{E}}})$ will result in an upper bound of $I(X;{Y_{\rm{B}}}|{Y_{\rm{E}}})$.
To obtain a tight upper bound, we have
\begin{equation}
I(X;{Y_{\rm{B}}}|{Y_{\rm{E}}}) = \mathop {\min }\limits_{{g_{{Y_{\rm{B}}}|{Y_{\rm E}}}}({y_{\rm{B}}}|{Y_{\rm{E}}})} {E_{X{Y_{\rm{E}}}}}\left[ {D\left( {\left. {{f_{{Y_{\rm{B}}}|X{Y_{\rm E}}}}({y_{\rm{B}}}|X,{Y_{\rm{E}}})} \right\|{g_{{Y_{\rm{B}}}|{Y_{\rm E}}}}({y_{\rm{B}}}|{Y_{\rm{E}}})} \right)} \right].
\label{eq17}
\end{equation}

From problem (\ref{eq6}), the secrecy capacity can be re-expressed as
\begin{equation}
{C_{\rm s}} = \mathop {\max }\limits_{{f_X}\left( x \right)} I\left( {X;{Y_{\rm{B}}}\left| {{Y_{\rm{E}}}} \right.} \right).
\label{eq18}
\end{equation}
In (\ref{eq18}), there exists a unique input PDF $f_{X^*}(x)$ that maximizes $I(X;{Y_{\rm{B}}}|{Y_{\rm{E}}})$ subject to the constraints in problem (\ref{eq6}).
Therefore, we have the dual expression of secrecy capacity as
\begin{eqnarray}
{C_{\rm s}} &=& \mathop {\max }\limits_{{f_X}(x)} \mathop {\min }\limits_{{g_{{Y_{\rm{B}}}|{Y_{\rm E}}}}({y_{\rm{B}}}|{Y_{\rm{E}}})} {E_{X{Y_{\rm{E}}}}}\left[ {D\left( {\left. {{f_{{Y_{\rm{B}}}|X{Y_{\rm E}}}}({y_{\rm{B}}}|X,{Y_{\rm{E}}})} \right\|{g_{{Y_{\rm{B}}}|{Y_{\rm E}}}}({y_{\rm{B}}}|{Y_{\rm{E}}})} \right)} \right] \nonumber\\
 &=& \mathop {\min }\limits_{{g_{{Y_{\rm{B}}}|{Y_{\rm E}}}}({y_{\rm{B}}}|{Y_{\rm{E}}})} {E_{{X^*}{Y_{\rm{E}}}}}\left[ {D\left( {\left. {{f_{{Y_{\rm{B}}}|X{Y_{\rm E}}}}({y_{\rm{B}}}|X,{Y_{\rm{E}}})} \right\|{g_{{Y_{\rm{B}}}|{Y_{\rm E}}}}({y_{\rm{B}}}|{Y_{\rm{E}}})} \right)} \right],
\label{eq19}
\end{eqnarray}
where ${X^*}$ denotes the optimal input.

To obtain a tight upper bound on secrecy capacity,
we should choose a tractable and suitable ${g_{{Y_{\rm{B}}}|{Y_{\rm E}}}}({y_{\rm{B}}}|{Y_{\rm{E}}})$.
Then, the following theorem is obtained.

\begin{theorem}
For indoor VLC having constraints (\ref{eq2}) and (\ref{eq4}), we derive an upper bound on the secrecy capacity as
\begin{eqnarray}
{C_{{\rm{Upp}}}} = \left\{ \begin{array}{l}
\ln \left( {\sqrt {\frac{{{\rm{4}}e{H_{\rm{E}}}\varsigma _{\rm{E}}^2\sigma _{\rm{E}}^2}}{{{\pi ^2}M}}} {\rm{ + }}\sqrt {\frac{{{\rm{2}}e\xi P{H_{\rm{B}}}{H_{\rm{E}}}\varsigma _{\rm{E}}^2\sigma _{\rm{E}}^2}}{{M\pi \varsigma _{\rm{B}}^2\sigma _{\rm{B}}^2}}} } \right),\; {\rm{if}}\;\frac{1}{{\sqrt {2\pi } }} \ge \frac{{{H_{\rm{E}}}}}{{{H_{\rm{B}}}}}\left( {\sqrt {\frac{{{H_{\rm{B}}}\varsigma _{\rm B}^2\sigma_{\rm B}^2}}{{2\pi M}}}  + \frac{{{H_{\rm{B}}}}}{2}\sqrt {\frac{{\xi P}}{M}} } \right)\\
\quad\quad\quad\quad\quad \frac{1}{2}\ln \left( {\frac{{4e{H_{\rm{B}}}\varsigma _{\rm{E}}^2\sigma _{\rm{E}}^2}}{{{\pi ^2}{H_{\rm{E}}}\varsigma _{\rm{B}}^2\sigma _{\rm{B}}^2}}} \right),\;\;\;\;\;\;\;\quad {\rm{otherwise}}
\end{array} \right.,
\label{eq20}
\end{eqnarray}
where $M = H_{\rm{E}}^2\varsigma _{\rm{B}}^2\sigma _{\rm{B}}^2/{H_{\rm{B}}} + {H_{\rm{E}}}\varsigma _{\rm{E}}^2\sigma _{\rm{E}}^2$.

\label{the2}
\end{theorem}

\begin{proof}
See Appendix \ref{Appendix_B}.
\end{proof}

\begin{corollary}
In \emph{Theorem \ref{the2}}, when we ignore the signal-dependent noise (i.e., ${\varsigma _{\rm{B}}} \to 0$ and ${\varsigma _{\rm{E}}} \to 0$), the secrecy capacity upper bound (\ref{eq20}) reduces to
\begin{eqnarray}
 \mathop {\lim }\limits_{{\varsigma _{\rm{B}}} \to 0\hfill\atop
{\varsigma _{\rm{E}}} \to 0\hfill} {C_{{\rm{Upp}}}} = \left\{ \begin{array}{l}
\ln \left[ {\frac{{4e\left( {\frac{{{\sigma _{\rm{B}}}}}{{\sqrt {2\pi } }} + \frac{{{H_{\rm{B}}}\xi P}}{2}} \right)}}{{\sqrt {2\pi e\sigma _{\rm{B}}^2\left( {1 + \frac{{H_{\rm{E}}^2\sigma _{\rm{B}}^2}}{{H_{\rm{B}}^2\sigma _{\rm{E}}^2}}} \right)} }}} \right],\;{\rm{if }}\,\frac{1}{{\sqrt {2\pi } }} \ge \frac{{{H_{\rm{E}}}}}{{\sqrt {H_{\rm{E}}^2\sigma _{\rm{B}}^2 + H_{\rm{B}}^2\sigma _{\rm{E}}^2} }}\left( {\frac{{{\sigma _{\rm{B}}}}}{{\sqrt {2\pi } }} + \frac{{{H_{\rm{B}}}\xi P}}{2}} \right)\\
\ln \left( {\frac{{2\sqrt e {H_{\rm{B}}}{\sigma _{\rm{E}}}}}{{\pi {H_{\rm{E}}}{\sigma _{\rm{B}}}}}} \right),\;\;\;\;\;\;\;\;\;\;\quad {\rm{otherwise}}
\end{array} \right..
\label{eq21}
\end{eqnarray}
\label{cor2}
\end{corollary}

\begin{proof}
The proof of \emph{Corollary \ref{cor2}} is straightforward and thus omitted here.
\end{proof}

\begin{remark}
Under constraints (\ref{eq2}) and (\ref{eq4}), a secrecy capacity upper bound (16) for VLC having only the signal-independent noise was also derived \cite{BIB28}.
As can be observed, the derived upper bound (\ref{eq21}) in \emph{Corollary \ref{cor2}} is the same as the upper bound (16) in \cite{BIB28}. This observation indicates that the result in \cite{BIB28} is just a special case of \emph{Theorem \ref{the2}} in this paper.
\label{rem2}
\end{remark}

\subsection{Asymptotic Behavior Analysis}
\label{section3_3}
In indoor VLC systems, we are more interested in the secrecy performance in the high optical intensity regime \cite{BIB16}.
Therefore, we let the nominal optical intensity of the LED $P$ tend to infinity.
According to \emph{Theorem \ref{the1}} and \emph{Theorem \ref{the2}}, the asymptotic performance is derived in the following corollary.

\begin{corollary}
For indoor VLC having constraints (\ref{eq2}) and (\ref{eq4}),
we derive asymptotic lower and upper bounds on secrecy capacity as
\begin{eqnarray}
\left\{ \begin{array}{l}
\displaystyle \mathop {\lim }\limits_{P \to \infty } {C_{{\rm{Low}}}} = \frac{1}{2}\ln \left( {\frac{{e{H_{\rm{B}}}\varsigma _{\rm{E}}^2\sigma _{\rm{E}}^2}}{{2\pi {H_{\rm{E}}}\varsigma _{\rm{B}}^2\sigma _{\rm{B}}^2}}} \right)\\
\displaystyle \mathop {\lim }\limits_{P \to \infty } {C_{{\rm{Upp}}}} = \frac{1}{2}\ln \left( {\frac{{4e{H_{\rm{B}}}\varsigma _{\rm{E}}^2\sigma _{\rm{E}}^2}}{{{\pi ^2}{H_{\rm{E}}}\varsigma _{\rm{B}}^2\sigma _{\rm{B}}^2}}} \right)
\end{array} \right..
\label{eq22}
\end{eqnarray}

\label{cor3}
\end{corollary}

\begin{proof}
See Appendix \ref{Appendix_C}.
\end{proof}

\begin{remark}
From \emph{Corollary \ref{cor3}}, both the asymptotic lower and upper bounds do not scale with the nominal optical intensity $P$ but converge to real and positive constants.
Moreover, the gap between the asymptotic upper bound and the asymptotic lower bound is given by
\begin{eqnarray}
{C_{{\rm{gap}}}} &=& \frac{1}{2}\ln \left( {\frac{{4e{H_{\rm{B}}}\varsigma _{\rm{E}}^2\sigma _{\rm{E}}^2}}{{{\pi ^2}{H_{\rm{E}}}\varsigma _{\rm{B}}^2\sigma _{\rm{B}}^2}}} \right) - \frac{1}{2}\ln \left( {\frac{{e{H_{\rm{B}}}\varsigma _{\rm{E}}^2\sigma _{\rm{E}}^2}}{{2\pi {H_{\rm{E}}}\varsigma _{\rm{B}}^2\sigma _{\rm{B}}^2}}} \right) \nonumber\\
 &=& \frac{1}{2}\ln \left( {\frac{8}{\pi }} \right) \approx 0.4674\;{\rm{nat/transmission}}{\rm{.}}
\label{eq23}
\end{eqnarray}
This indicates that the asymptotic lower and upper bounds on secrecy capacity do not coincide,
but the asymptotic performance gap is 0.4674 nat/transmission, which is small.
\label{rem3}
\end{remark}

\section{Secrecy Capacity for VLC Having Non-negativity, Average Optical Intensity and Peak Optical Intensity Constraints}
\label{section4}
In this section, we further derive exact and asymptotic secrecy capacity bounds for the VLC system when an additional peak optical intensity constraint is imposed on the channel input.

By considering constraints (\ref{eq2}), (\ref{eq3}) and (\ref{eq4}), we can derive the secrecy capacity by solving the following functional problem, i.e.,
\begin{eqnarray}
&&{C_{\rm s}} = \mathop {\max }\limits_{{f_X}\left( x \right)} \left[ {I\left( {X;{Y_{\rm{B}}}} \right) - I\left( {X;{Y_{\rm{E}}}} \right)} \right] \nonumber\\
{\rm{s}}{\rm{.t}}{\rm{.}}\;&&\int_0^A {{f_X}\left( x \right){\rm{d}}x}  = 1  \label{eq24}\\
&&E\left( X \right) = \int_0^A {x{f_X}\left( x \right){\rm{d}}x}  = \xi P. \nonumber
\end{eqnarray}
Similar to problem (\ref{eq6}), it is challenging to obtain the exact secrecy capacity expression for problem (\ref{eq24}).
We will derive tight lower and upper bounds on secrecy capacity in the following two subsections.

\subsection{Lower Bound on Secrecy Capacity}
\label{section4_1}
By choosing an arbitrary input PDF in problem (\ref{eq24}) satisfying constraints (\ref{eq2}), (\ref{eq3}) and (\ref{eq4}),
we can obtain a lower bound on the secrecy capacity as
\begin{eqnarray}
{C_{\rm s}} &\ge& {\left. {\left[ {I\left( {X\;;{Y_{\rm{B}}}} \right) - I\left( {X\;;{Y_{\rm{E}}}} \right)} \right]} \right|_{\forall {f_X}\left( x \right)\;{\rm{satifies}}\;(2),\,\;(3)\;{\rm and}\;(4)}} \nonumber\\
 &=& {\cal H}\left( {{Y_{\rm{B}}}} \right) - {\cal H}\left( {{Y_{\rm{B}}}\left| X \right.} \right) - {\cal H}\left( {{Y_{\rm{E}}}} \right){\rm{ + }}{\cal H}\left( {{Y_{\rm{E}}}\left| X \right.} \right).
\label{eq25}
\end{eqnarray}

In this case, the secrecy capacity lower bound (\ref{eq13}) also holds.
Define the average to peak optical intensity ratio as $\alpha  = {{\xi P} \mathord{\left/
 {\vphantom {{\xi P} A}} \right.
 \kern-\nulldelimiterspace} A}$, we derive a lower bound on secrecy capacity in the following theorem.

\begin{theorem}
For VLC having constraints (\ref{eq2}), (\ref{eq3}) and (\ref{eq4}),
we derive a lower bound on the secrecy capacity as
\begin{eqnarray}
{C'_{{\rm{Low}}}} = \left\{ \begin{array}{l}
{C_1},\,{\rm{if}}\,\alpha  = 0.5\\
{C_2},\,{\rm{if}}\,\alpha  \ne 0.5
\end{array} \right.,
\label{eq26}
\end{eqnarray}
where ${C_1}$ and ${C_2}$ are defined as
\begin{eqnarray}
{C_1} &=& {f_{{\rm{low}}}}\left( {{H_{\rm{B}}},\xi ,P} \right) + \frac{1}{2}\ln \left( {\frac{{6{A^2}\sigma _{\rm{E}}^2}}{{\pi e\sigma _{\rm{B}}^2\left( {H_{\rm{E}}^2{A^2} + 6A{H_{\rm{E}}}\varsigma _{\rm{E}}^2\sigma _{\rm{E}}^2 + 12\sigma _{\rm{E}}^2} \right)}}} \right) \nonumber\\
 &&+ \frac{1}{2}\ln \left( {\frac{{1 + {H_{\rm{E}}}\varsigma _{\rm{E}}^2A}}{{1 + {H_{\rm{B}}}\varsigma _{\rm{B}}^2A}}} \right) - \frac{{\ln \left( {1 + {H_{\rm{B}}}\varsigma _{\rm{B}}^2A} \right)}}{{2A{H_{\rm{B}}}\varsigma _{\rm{B}}^2}} + \frac{{\ln \left( {1 + {H_{\rm{E}}}\varsigma _{\rm{E}}^2A} \right)}}{{2A{H_{\rm{E}}}\varsigma _{\rm{E}}^2}},
\label{eq27}
\end{eqnarray}
\begin{eqnarray}
{C_2} \!\!\!&=&\!\!\! {f_{{\rm{low}}}}\left(\! {{H_{\rm{B}}},\xi ,P} \!\right) \!-\! c\xi P \!+\! \frac{1}{2}\ln\!\! \left(\! {\frac{{\sigma _{\rm{E}}^2}}{{\sigma _{\rm{B}}^2}}\frac{{{{({e^{cA}} - 1)}^2}}}{{{c^2}}}} \!\right) \!+\! \frac{1}{{2\left( {{e^{cA}} - 1} \right)}}\!\left\{\! {\ln \!\!\left(\! {\frac{{1 + {H_{\rm{E}}}\varsigma _{\rm{E}}^2A}}{{1 + {H_{\rm{B}}}\varsigma _{\rm{B}}^2A}}} \right){e^{cA}}} \right. \nonumber\\
 &&- {e^{ - \frac{c}{{{H_{\rm{E}}}\varsigma _{\rm{E}}^2}}}}\left[ {Ei\left( {\frac{c}{{{H_{\rm{E}}}\varsigma _{\rm{E}}^2}}\left( {1 + {H_{\rm{E}}}\varsigma _{\rm{E}}^2A} \right)} \right) - Ei\left( {\frac{c}{{{H_{\rm{E}}}\varsigma _{\rm{E}}^2}}} \right)} \right] \nonumber\\
&&\left. { + {e^{ - \frac{c}{{{H_{\rm{B}}}\varsigma _{\rm{B}}^2}}}}\left[ {Ei\left( {\frac{c}{{{H_{\rm{B}}}\varsigma _{\rm{B}}^2}}\left( {1 + {H_{\rm{B}}}\varsigma _{\rm{B}}^2A} \right)} \right) - Ei\left( {\frac{c}{{{H_{\rm{B}}}\varsigma _{\rm{B}}^2}}} \right)} \right]} \right\} \nonumber\\
 &&- \frac{1}{2}\ln \left[ {2\pi e\left( {H_{\rm{E}}^2\left( {\frac{{A\left( {cA - 2} \right)}}{{c\left( {1 - {e^{ - cA}}} \right)}} + \frac{2}{{{c^2}}} - {\xi ^2}{P^2}} \right) + {H_{\rm{E}}}\xi P\varsigma _{\rm{E}}^2\sigma _{\rm{E}}^2 + \sigma _{\rm{E}}^2} \right)} \right],
\label{eq28}
\end{eqnarray}
where $c$ in (\ref{eq28}) is the solution to the following equation
\begin{equation}
\alpha  = \frac{1}{{1 - {e^{ - cA}}}} - \frac{1}{{cA}}.
\label{eq29}
\end{equation}

\label{the3}
\end{theorem}

\begin{proof}
See Appendix \ref{Appendix_D}.
\end{proof}

From Appendix \ref{Appendix_D}, the maxentropic input PDF is used to obtain the lower bound of the secrecy capacity in \emph{Theorem \ref{the3}}.
For different $\alpha $ values, the maxentropic input PDFs are different.
The properties of such PDFs are provided in the following theorem.

\begin{theorem}
If $\alpha  = 0.5$, then the maxentropic input PDF (\ref{eq75}) is uniformly distributed in $[0,A]$;
if $\alpha  \in (0.5,1]$, then $c > 0$ in (\ref{eq29}), and the maxentropic input PDF (\ref{eq79}) is a monotonically increasing function of $x \in (0,A]$;
if $\alpha  \in (0,0.5)$, then $c < 0$ in (\ref{eq29}) and the maxentropic input PDF (\ref{eq79}) is a monotonically decreasing function of $x \in (0,A]$.
Moreover, the curves of maxentropic input PDF (\ref{eq79}) with $\alpha $ and $1 - \alpha $ are symmetric with respect to $X = A/2$.
\label{the4}
\end{theorem}

\begin{proof}
See Appendix \ref{Appendix_E}.
\end{proof}

From \emph{Theorem \ref{the3}}, the following corollary and remark can be derived.

\begin{corollary}
In \emph{Theorem \ref{the3}}, when we ignore the signal-dependent noise (i.e., ${\varsigma _{\rm{B}}} \to 0$ and ${\varsigma _{\rm{E}}} \to 0$), the secrecy capacity lower bound (\ref{eq26}) reduces to
\begin{eqnarray}
\mathop {\lim }\limits_{\varsigma _{\rm{B}} \to 0\hfill\atop
\varsigma _{\rm{E}} \to 0\hfill} {C'_{{\rm{Low}}}} = \left\{ \begin{array}{l}
\frac{1}{2}\ln \left[ {\frac{{3H_{\rm{B}}^2\sigma _{\rm{E}}^2{A^2}}}{{2\pi e\sigma _{\rm{B}}^2\left( {{\xi ^2}{P^2}H_{\rm{E}}^2 + 3\sigma _{\rm{E}}^2} \right)}}} \right],\;\;\;\;\;\;\;\;\;\;\;\;\;\;\;\;\;\;\;\;\;\;{\rm{if}}\,\alpha  = 0.5\\
\frac{1}{2}\ln \left[ {\frac{{\sigma _{\rm{E}}^2H_{\rm{B}}^2{e^{ - 2c\xi P}}{{\left( {\frac{{{e^{cA}} - 1}}{c}} \right)}^2}}}{{2\pi e\sigma _{\rm{B}}^2\left( {\frac{{H_{\rm{E}}^2A\left( {cA - 2} \right)}}{{c\left( {1 - {e^{ - cA}}} \right)}} + \frac{{2H_{\rm{E}}^2}}{{{c^2}}} - H_{\rm{E}}^2{\xi ^2}{P^2} + \sigma _{\rm{E}}^2} \right)}}} \right],\,{\rm{if}}\,\alpha  \ne 0.5
\end{array} \right..
\label{eq30}
\end{eqnarray}
\label{cor4}
\end{corollary}

\begin{proof}
According to \emph{Theorem \ref{the3}}, the proof of \emph{Corollary \ref{cor4}} is straightforward.
\end{proof}

\begin{remark}
Under constraints (\ref{eq2}), (\ref{eq3}) and (\ref{eq4}),
a secrecy capacity lower bound (20) for VLC having only the signal-independent noise was derived \cite{BIB28}.
As can be seen, the derived lower bound (\ref{eq30}) in \emph{Corollary \ref{cor4}} is smaller than the lower bound (20) in \cite{BIB28}.
This is because the principle of OELIE is employed in \emph{Corollary \ref{cor4}}, while the EPI is used in \cite{BIB28}.
This indicates that the EPI approach is more efficient to analyze the secrecy capacity lower bound for VLC without signal-dependent noise.
\label{rem4}
\end{remark}

\subsection{Upper Bound on Secrecy Capacity}
\label{section4_2}
When adding an additional peak optical intensity constraint, the secrecy capacity's dual expression (\ref{eq19}) also holds.
According to (\ref{eq19}) and \emph{Theorem \ref{the2}}, we have the following theorem.

\begin{theorem}
For VLC having constraints (\ref{eq2}), (\ref{eq3}) and (\ref{eq4}), we derive an upper bound on secrecy capacity as
\begin{equation}
{C'_{{\rm{Upp}}}} = \frac{1}{2}\ln \left[ {\frac{{{H_{\rm{E}}}\varsigma _{\rm{E}}^2\sigma _{\rm{E}}^2\left( {{H_{\rm{B}}}A + \varsigma _{\rm{B}}^2\sigma _{\rm{B}}^2} \right)}}{{\varsigma _{\rm{B}}^2\sigma _{\rm{B}}^2\left( {H_{\rm{E}}^2A + \frac{{H_{\rm{E}}^2}}{{{H_{\rm{B}}}}}\varsigma _{\rm{B}}^2\sigma _{\rm{B}}^2 + M} \right)}}} \right],
\label{eq31}
\end{equation}
where $M = H_{\rm{E}}^2\varsigma _{\rm{B}}^2\sigma _{\rm{B}}^2/{H_{\rm{B}}} + {H_{\rm{E}}}\varsigma _{\rm{E}}^2\sigma _{\rm{E}}^2$.

\label{the5}
\end{theorem}

\begin{proof}
See Appendix \ref{Appendix_F}.
\end{proof}

\begin{corollary}
In \emph{Theorem \ref{the5}}, when we ignore the signal-dependent noise (i.e., ${\varsigma _{\rm{B}}} \to 0$ and ${\varsigma _{\rm{E}}} \to 0$), the secrecy capacity upper bound (\ref{eq31}) reduces to
\begin{equation}
\mathop {\lim }\limits_{\scriptstyle{\varsigma _{\rm{B}}} \to 0\hfill\atop
\scriptstyle{\varsigma _{\rm{E}}} \to 0\hfill} {C'_{{\rm{Upp}}}} = \frac{1}{2}\ln \left[ {\frac{{\left( {H_{\rm{B}}^2A\xi P + \sigma _{\rm{B}}^2} \right)\sigma _{\rm{E}}^2}}{{\left( {H_{\rm{E}}^2A\xi P + 2\frac{{H_{\rm{E}}^2}}{{H_{\rm{B}}^2}}\sigma _{\rm{B}}^2 + \sigma _{\rm{E}}^2} \right)\sigma _{\rm{B}}^2}}} \right].
\label{eq32}
\end{equation}

\label{cor5}
\end{corollary}

\begin{remark}
Under constraints (\ref{eq2}), (\ref{eq3}) and (\ref{eq4}),
a secrecy capacity upper bound (26) for VLC without signal-dependent noise was derived \cite{BIB28}.
As can be seen, the derived upper bound (\ref{eq32}) in \emph{Corollary \ref{cor5}} is the same as that (26) in \cite{BIB28}.
This also indicates that the result of \cite{BIB28} is just a special case of \emph{Theorem \ref{the5}} in this paper.
\label{rem5}
\end{remark}

\subsection{Asymptotic Behavior Analysis}
\label{section4_3}
According to \emph{Theorem \ref{the3}} and \emph{Theorem \ref{the5}},
when the peak optical intensity of the LED $A$ tends to infinity,
we can get the asymptotic secrecy capacity bounds.

\begin{corollary}
For VLC having constraints (\ref{eq2}), (\ref{eq3}) and (\ref{eq4}),
we derive asymptotic lower and upper bounds on secrecy capacity as
\begin{eqnarray}
\left\{ \begin{array}{l}
\mathop {\lim }\limits_{A \to \infty }\!\! {C'_{{\rm{Low}}}} \!=\!\! \left\{ \begin{array}{l}
\frac{1}{2}\ln \left( {\frac{{6{H_{\rm{B}}}\varsigma _{\rm{E}}^2\sigma _{\rm{E}}^2}}{{\pi e{H_{\rm{E}}}\varsigma _{\rm{B}}^2\sigma _{\rm{B}}^2}}} \right),\,{\rm{if}}\,\alpha  = 0.5\\
\!\!\mathop {\lim }\limits_{A \to \infty } \!\frac{1}{2}\!\ln\!\!\! \left[ {\frac{{{H_{\rm{B}}}{H_{\rm{E}}}\varsigma _{\rm{E}}^2\sigma _{\rm{E}}^2{{({e^{cA}} - 1)}^2}}}{{2\pi e{c^2}\varsigma _{\rm{B}}^2\sigma _{\rm{B}}^2{e^{2c\xi P}}\left[ {H_{\rm{E}}^2\left( {\frac{{A\left( {cA - 2} \right)}}{{c\left( {1 - {e^{ - cA}}} \right)}} + \frac{2}{{{c^2}}} - {\xi ^2}{P^2}} \right) + {H_{\rm{E}}}\xi P\varsigma _{\rm{E}}^2\sigma _{\rm{E}}^2 + \sigma _{\rm{E}}^2} \right]}}} \!\!\right]\!,\;\,{\rm{if}}\,\; \alpha  \ne 0.5
\end{array} \right.\\
\mathop {\lim }\limits_{A \to \infty } {C'_{{\rm{Upp}}}} = \frac{1}{2}\ln \left( {\frac{{{H_{\rm{B}}}\varsigma _{\rm{E}}^2\sigma _{\rm{E}}^2}}{{{H_{\rm{E}}}\varsigma _{\rm{B}}^2\sigma _{\rm{B}}^2}}} \right).
\end{array} \right.
\label{eq33}
\end{eqnarray}
\label{cor6}
\end{corollary}

\begin{proof}
See Appendix \ref{Appendix_G}.
\end{proof}

\begin{remark}
In \emph{Corollary \ref{cor6}}, when $\alpha  = 0.5$, both the asymptotic lower and upper bounds do not scale with $A$ but converge to real and positive constants.
Specifically, the gap between the asymptotic upper bound and the asymptotic lower bound is given by
\begin{eqnarray}
{C'_{{\rm{gap}}}} &=& \frac{1}{2}\ln \left( {\frac{{{H_{\rm{B}}}\varsigma _{\rm{E}}^2\sigma _{\rm{E}}^2}}{{{H_{\rm{E}}}\varsigma _{\rm{B}}^2\sigma _{\rm{B}}^2}}} \right) - \frac{1}{2}\ln \left( {\frac{{6{H_{\rm{B}}}\varsigma _{\rm{E}}^2\sigma _{\rm{E}}^2}}{{\pi e{H_{\rm{E}}}\varsigma _{\rm{B}}^2\sigma _{\rm{B}}^2}}} \right) \nonumber\\
 &=& \frac{1}{2}\ln \left( {\frac{{\pi e}}{6}} \right) \approx 0.1765\,{\rm{nat/transmission}}{\rm{.}}
\label{eq34}
\end{eqnarray}
This indicates that the asymptotic performance gap is small.

When $\alpha  \ne 0.5$, it is challenging to obtain an exact asymptotic lower bound on the secrecy capacity,
and thus the performance gap between the upper and lower bounds on secrecy capacity can only be evaluated by using numerical results in Section \ref{section5}.
\label{rem6}
\end{remark}

\section{Numerical Results}
\label{section5}
In this section, some classic numerical results will be provided.
The accuracy of the derived secrecy capacity bounds will also be verified.
In the simulation, we assume that $\sigma _{\rm{B}}^2 = \sigma _{\rm{E}}^2 = 1$ and $\varsigma _{\rm{B}}^2 = \varsigma _{\rm{E}}^2 = 1.5$.

\subsection{Results of VLC Having Non-negativity and Average Optical Intensity Constraints}
To verify the accuracy of the lower bound (\ref{eq14}) and the upper bound (\ref{eq20}), we provide Fig. \ref{fig2}-Fig. \ref{fig5} as well as Table \ref{tab1} in this subsection.

Fig. \ref{fig2} shows the secrecy capacity bounds versus the nominal optical intensity $P$ with different $H_{\rm B}/H_{\rm E}$ when $\xi  = 0.3$.
In the low optical intensity regime, the secrecy capacity bounds approach zero, and the PLS performance is bad.
However, in the high optical intensity regime, the secrecy capacity bounds increase rapidly first and then tend to stable values with $P$.
Moreover, for a fixed $P$, the secrecy capacity bounds increase with $H_{\rm B}/H_{\rm E}$.
This indicates that the larger the difference between ${H_{\rm{B}}}$ and ${H_{\rm{E}}}$ is, the PLS performance becomes better.
To further quantitative the differences between the upper and lower bounds on secrecy capacity in the high optical intensity regime, we show the performance gaps for different $H_{\rm B}/H_{\rm E}$ scenarios in Table \ref{tab1}.
It can be observed that all performance gaps for different $H_{\rm B}/H_{\rm E}$ scenarios are about 0.4674 nat/transmission, which indicates that the asymptotic secrecy capacity upper and lower bounds are tight.
This conclusion coincides with that in \emph{Remark 3}.

\begin{figure}
\centering
\rotatebox[origin=c]{-90}{\includegraphics[width=9cm]{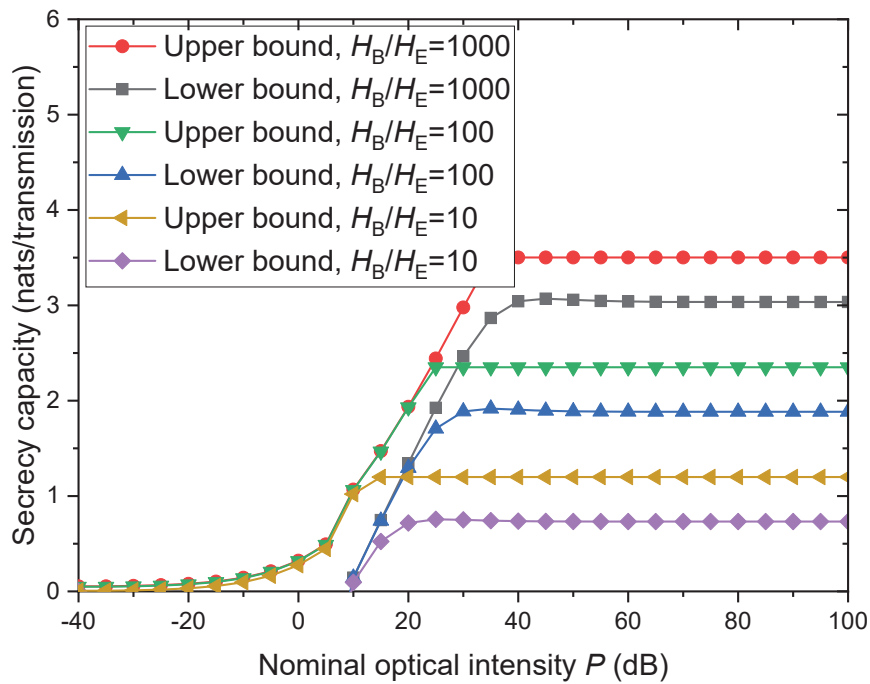}}
\caption{Secrecy capacity bounds versus nominal optical intensity $P$ with different ${H_{\rm{B}}}/{H_{\rm{E}}}$ when $\xi  = 0.3$}
\label{fig2}
\end{figure}

\begin{table}[!h]
\renewcommand{\arraystretch}{1}
\caption{Performance gaps between (\ref{eq14}) and (\ref{eq20}) at high optical intensity in Fig. \ref{fig2}}
\label{tab1}
\centering
\begin{tabular}{|c|c|c|c|}
\hline\hline
\textbf{P (dB)} & \bm{${H_{\rm{B}}}/{H_{\rm{E}}} = 1000$} & \bm{${H_{\rm{B}}}/{H_{\rm{E}}} = 100$} & \bm{${H_{\rm{B}}}/{H_{\rm{E}}} = 10$}\\
\hline
85 & 0.4673 & 0.4674 & 0.4674 \\
\hline
90 & 0.4673 & 0.4674 & 0.4674 \\
\hline
95 & 0.4674 & 0.4674 & 0.4674 \\
\hline
100& 0.4674 & 0.4674 & 0.4674 \\
\hline\hline
\end{tabular}
\end{table}

\begin{figure}
\centering
\rotatebox[origin=c]{-90}{\includegraphics[width=9cm]{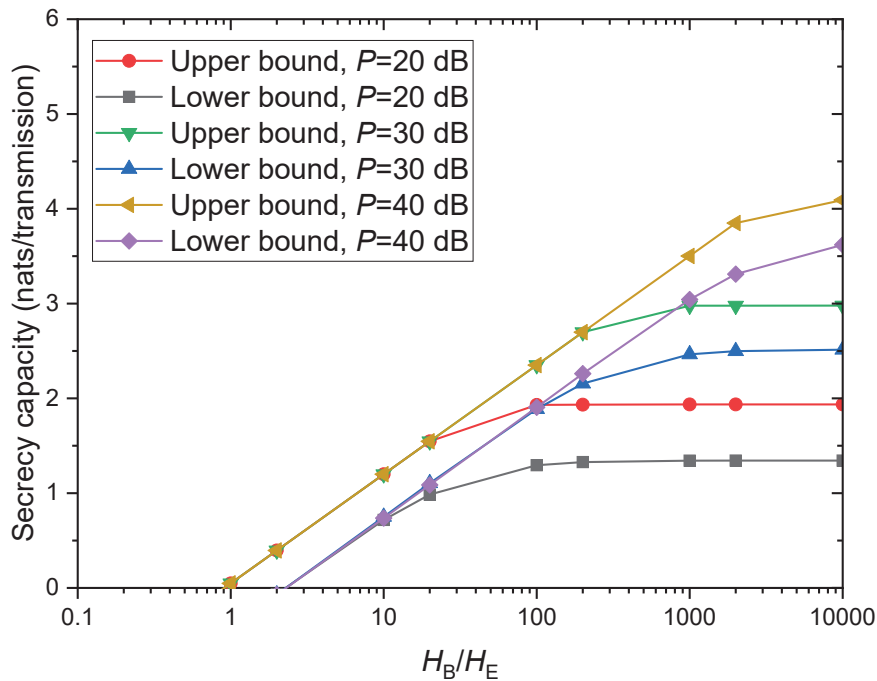}}
\caption{Secrecy capacity bounds versus ${H_{\rm{B}}}/{H_{\rm{E}}}$\emph{} with different $P$ when $\xi  = 0.3$ \label{fig3}}
\end{figure}

Fig. \ref{fig3} plots the secrecy capacity bounds versus $H_{\rm B}/H_{\rm E}$ with different $P$ when $\xi  = 0.3$.
As can be seen, when $H_{\rm B}/H_{\rm E} \le 1$,
the main channel is worse than the eavesdropping channel, the secrecy capacities are zero.
When $H_{\rm B}/H_{\rm E} > 1$,
the main channel outperforms the eavesdropping channel, and the secrecy capacity bounds increase first and then tend to stable values with the increase of $H_{\rm B}/H_{\rm E}$.
For small $H_{\rm B}/H_{\rm E}$ values (for example, when $1 \le H_{\rm B}/H_{\rm E} \le 10$),
all lower bounds (or all upper bounds) coincide with each other.
In this case, the secrecy capacity cannot be enhanced by increasing the nominal optical intensity.
However, for large $H_{\rm B}/H_{\rm E}$ values,
the secrecy performance can be improved by enlarging the nominal optical intensity.
Moreover, for large $H_{\rm B}/H_{\rm E}$ values (for examples, $H_{\rm B}/H_{\rm E} > 100$ when $P = 20\,{\rm{dB}}$; $H_{\rm B}/H_{\rm E} > 2000$ when $P = 30\,{\rm{dB}}$),
Eve can almost not able to eavesdrop information due to the small channel gain of the eavesdropping channel,
and thus the stable secrecy capacity bounds in this case can be approximated as the channel capacity bounds between Alice and Bob.

Fig. \ref{fig4} plots the secrecy capacity bounds versus the dimming target $\xi $ with different $P$ when $H_{\rm B}/H_{\rm E} = 1000$.
As can be seen, all secrecy capacity bounds are monotonically non-decreasing functions with respect to $\xi $.
At small optical intensity (i.e., $P = 20\,{\rm{dB}}$), the secrecy capacity bounds increase rapidly as the increase of $\xi $.
However, in large optical intensity regime (i.e., $P = 40\,{\rm{dB}}$ and $P = 60\,{\rm{dB}}$), the secrecy capacity bounds increase first and then tends to constant values.
This indicates that large dimming target has a strong impact on PLS performance when $P$ is small, but has a weak impact on PLS performance when $P$ is large.
Moreover, it can be observed that the capacity bounds when $P = 40\,{\rm{dB}}$ are almost the same as that when $P = 60\,{\rm{dB}}$.
This indicates that for fixed $H_{\rm B}/H_{\rm E}$ and $\xi $ values,
enlarging the optical intensity cannot enhance the secrecy performance of VLC without limitation at all.

\begin{figure}
\centering
\rotatebox[origin=c]{-90}{\includegraphics[width=9cm]{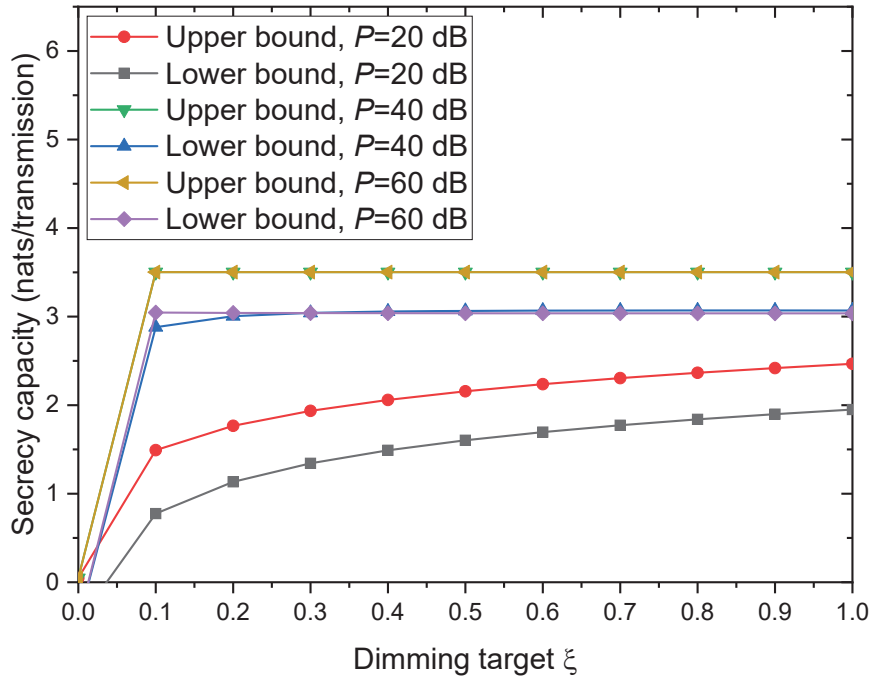}}
\caption{Secrecy capacity bounds versus dimming target $\xi $ with different $P$ when ${H_{\rm{B}}}/{H_{\rm{E}}} =1000$ \label{fig4}}
\end{figure}

Fig. \ref{fig5} shows the effects of the noise and the eavesdropper on capacity performance when $\xi  = 0.3$ and $H_{\rm B}/H_{\rm E} = 1000$.
In this figure, the secrecy capacity bounds for VLC having the signal-dependent noise (i.e., lower bound (\ref{eq14}) and upper bound (\ref{eq20}) in this paper),
the secrecy capacity bounds for VLC having the signal-independent noise (i.e., lower bound (8) and upper bound (16) in \cite{BIB28}),
and the channel capacity bounds for VLC having the signal-dependent noise (i.e., lower bound (33) and upper bound (35) in \cite{BIB16}),
the channel capacity bounds for VLC having the signal-independent noise (i.e., lower bound (37) and upper bound (28) in \cite{BIB04}) are provided.
It can be observed that all secrecy capacity bounds increase and then tend to stable values as the increase of $P$.
However, all channel capacity bounds monotonously increase with $P$. Moreover, for a fixed type of noise, the channel capacity bounds always larger than the secrecy capacity bounds.
This indicates that the existence of the eavesdropper in the system degrade the information transmission ability.
Furthermore, compared with the signal-independent noise, the signal-dependent noise will decrease the channel capacity and the secrecy capacity of the VLC system.

\begin{figure}
\centering
\rotatebox[origin=c]{-90}{\includegraphics[width=9cm]{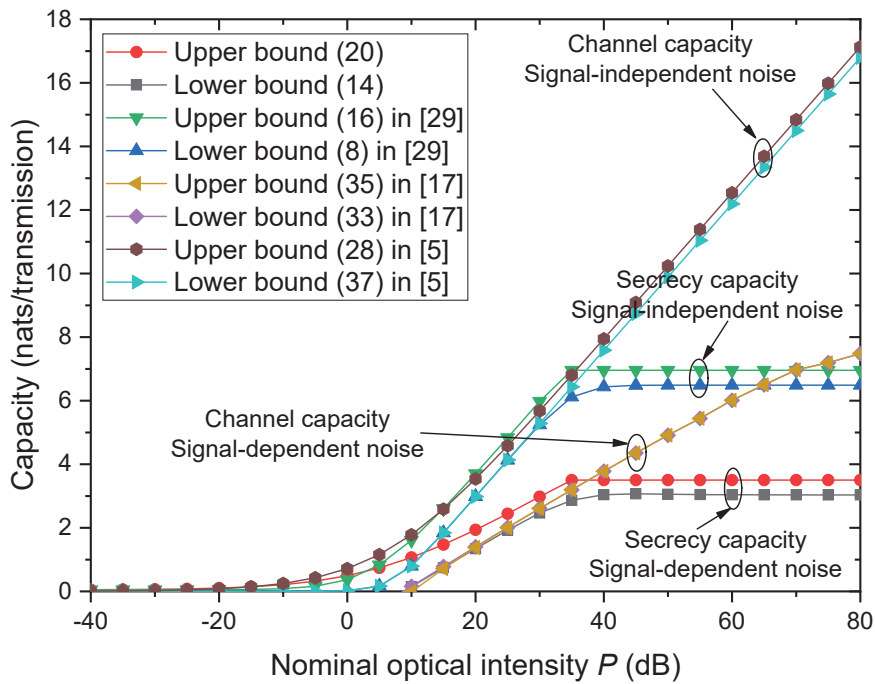}}
\caption{The effects of the noise and the eavesdropper on capacity performance when $\xi  = 0.3$ and ${H_{\rm{B}}}/{H_{\rm{E}}} =1000$. \label{fig5}}
\end{figure}

\subsection{Results of VLC Having Non-negativity, Average Optical Intensity and Peak Optical Intensity Constraints}
To verify \emph{Theorem \ref{the4}}, Fig. 6 shows the maxentropic input PDFs for different $\alpha$ values when $A = {10^6}\,{\rm{W}}$.
It can be seen that the curve of the PDF when $\alpha  = 0.5$ does not vary with the input $X$.
In other words, the value of the probability for each point is equal to $1/A$, which indicates that the input follows a uniform distribution in $[0,A]$, i.e., (\ref{eq75}).
When $\alpha  < 0.5$, the PDF is a monotonically decreasing function of $x \in (0,A]$.
However, when $\alpha  > 0.5$, the PDF becomes a monotonically increasing function of $x \in (0,A]$.
Moreover, it can be observed that the PDF curves with $\alpha $ and $1 - \alpha $ are symmetric with respect to $X = A/2$.
Therefore, all conclusions in \emph{Theorem \ref{the4}} hold.

\begin{figure}
\centering
\rotatebox[origin=c]{-90}{\includegraphics[width=9cm]{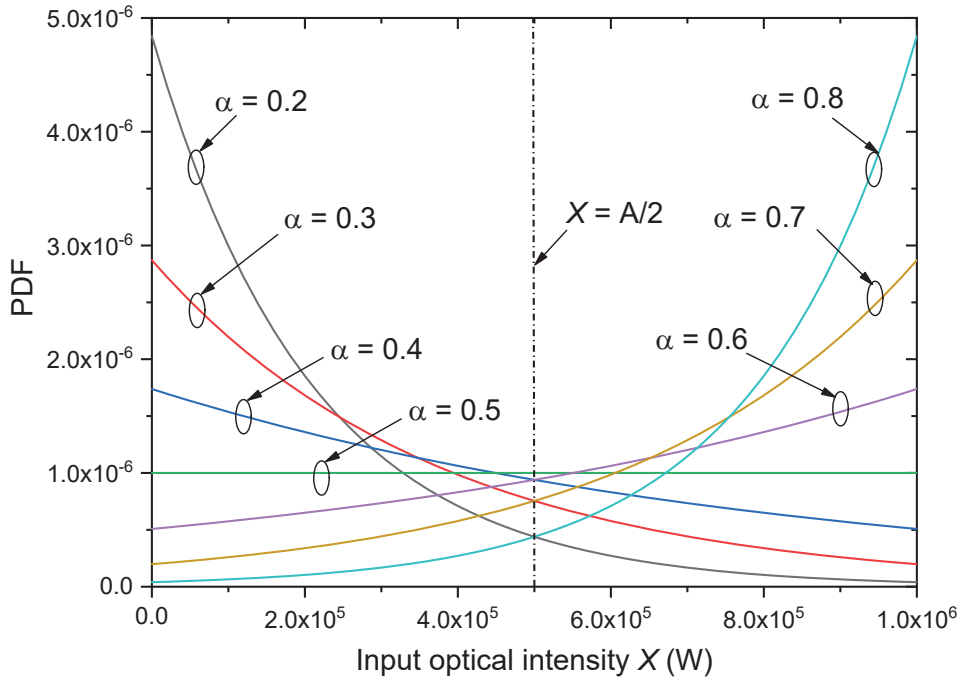}}
\caption{The maxentropic input PDFs for different $\alpha $ values when $A = {10^6}\,{\rm{W}}$ \label{fig6}}
\end{figure}

To verify the accuracy of the lower bound (\ref{eq26}) and the upper bound (\ref{eq31}), we plot Fig. \ref{fig7}-Fig. \ref{fig9} and Table \ref{tab2} in this subsection.

\begin{figure}
    \centering
    \subfigure[$\alpha  = 0.2$ (i.e., $\xi  = 0.3,\,A/P = 1.5$)]{
    \begin{minipage}{10cm}
    \centering
        \rotatebox[origin=c]{-90}{\includegraphics[width=8cm]{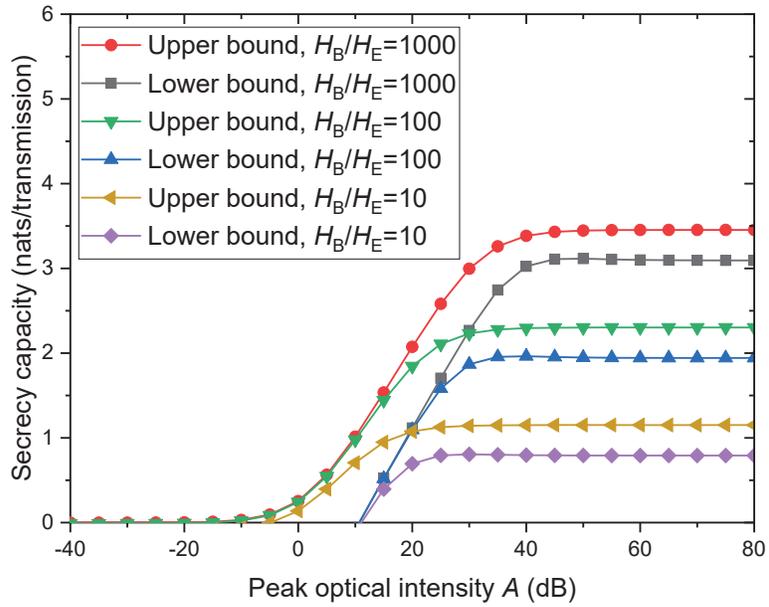}}
    \end{minipage}
    }
    \subfigure[$\alpha  = 0.5$ (i.e., $\xi  = 0.5,\,A/P = 1$)]{
    \begin{minipage}{10cm}
     \centering
        \rotatebox[origin=c]{-90}{\includegraphics[width=8cm]{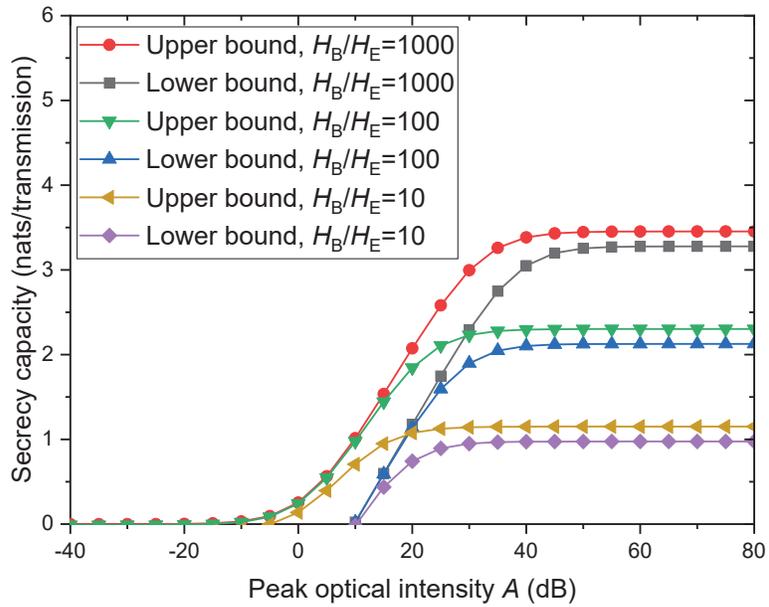}}
    \end{minipage}
    }
    \caption{Secrecy capacity bounds versus peak optical intensity $A$ with different ${H_{\rm{B}}}/{H_{\rm{E}}}$}
    \label{fig7}
\end{figure}

Fig. \ref{fig7}(a) and (b) plot the secrecy capacity bounds versus the peak optical intensity $A$ with different ${H_{\rm{B}}}/{H_{\rm{E}}}$ when $\alpha  = 0.2$ and 0.5.
Similar to Fig. \ref{fig2}, for a fixed $A$, the secrecy capacity bounds increase with ${H_{\rm{B}}}/{H_{\rm{E}}}$.
This indicates that the larger the difference between the main channel and the eavesdropping channel is, the better the secrecy capacity performance becomes.
Moreover, when $A$ is small (for example, $A \le  - 10\;{\rm{dB}}$), the secrecy capacity bounds are almost zero.
As the increase of $A$, the secrecy capacity bounds increase rapidly. When $A$ is large (for examples, $A \ge 30\,{\rm{dB}}$ when ${H_{\rm{B}}}/{H_{\rm{E}}} = 10$; $A \ge 40\,{\rm{dB}}$ when ${H_{\rm{B}}}/{H_{\rm{E}}} = 100$; and $A \ge 50\,{\rm{dB}}$ when ${H_{\rm{B}}}/{H_{\rm{E}}} = 1000$),
all secrecy capacity bounds tend to constants, which indicates that the secrecy capacity bounds are not affected by $A$ in large optical intensity regime.
Furthermore, Table \ref{tab2} shows the performance gaps between the upper and lower bounds on secrecy capacity in large optical intensity regime.
As it is seen, when $\alpha  = 0.2$, the performance gaps for different ${H_{\rm{B}}}/{H_{\rm{E}}}$ are about 0.36 nat/transmission, which is small.
When $\alpha  = 0.5$, the gaps are about 0.1765 nat/transmission, which agrees with \emph{Remark \ref{rem6}}.

\begin{table}[]
\caption{Performance gaps between the lower bound (\ref{eq26}) and the upper bound (\ref{eq31}) at high optical intensity in Fig. \ref{fig7}}
\label{tab2}
\centering
\begin{tabular}{|c|c|c|c|c|}
\hline\hline
\bm{$\alpha$}                & \textbf{A (dB)} & \bm{${H_{\rm{B}}}/{H_{\rm{E}}} = 1000$}   & \bm{${H_{\rm{B}}}/{H_{\rm{E}}} = 100$}    & \bm{${H_{\rm{B}}}/{H_{\rm{E}}} = 10$}     \\ \hline
\multirow{4}{*}{0.2} & 65     & 0.3574 & 0.3596 & 0.3600   \\ \cline{2-5}
                     & 70     & 0.359  & 0.3599 & 0.3600   \\ \cline{2-5}
                     & 75     & 0.3596 & 0.3599 & 0.3600   \\ \cline{2-5}
                     & 80     & 0.3599 & 0.3600 & 0.3600   \\ \hline
\multirow{4}{*}{0.5} & 65     & 0.1767 & 0.1765 & 0.1765 \\ \cline{2-5}
                     & 70     & 0.1765 & 0.1765 & 0.1765 \\ \cline{2-5}
                     & 75     & 0.1765 & 0.1765 & 0.1765 \\ \cline{2-5}
                     & 80     & 0.1765 & 0.1765 & 0.1765 \\ \hline\hline
\end{tabular}
\end{table}

Fig. \ref{fig8} plots the secrecy capacity bounds versus ${H_{\rm{B}}}/{H_{\rm{E}}}$ with different $A$ when $\xi  = 0.3$ and $A/P = 1.5$.
Similar to Fig. \ref{fig3}, we pay more attention to the results when ${H_{\rm{B}}}/{H_{\rm{E}}} > 1$.
When the main channel is better than the eavesdropping channel (i.e., ${H_{\rm{B}}}/{H_{\rm{E}}} > 1$),
all secrecy capacity bounds increase fast and then tend to stable values with ${H_{\rm{B}}}/{H_{\rm{E}}}$.
When ${H_{\rm{B}}}/{H_{\rm{E}}}$ is sufficiently large (for example, ${H_{\rm{B}}}/{H_{\rm{E}}} > 1000$ when $A = 20\,{\rm{dB}}$),
Eve can almost not able to eavesdrop information, the secrecy capacity bounds approximately reduce to the channel capacity bounds, which do not vary with ${H_{\rm{B}}}/{H_{\rm{E}}}$.
Moreover, in the large optical intensity regime, the secrecy performance improves with an increase of $A$.
These aforementioned conclusions are similar to that in Fig. \ref{fig3}.

\begin{figure}
\centering
\rotatebox[origin=c]{-90}{\includegraphics[width=9cm]{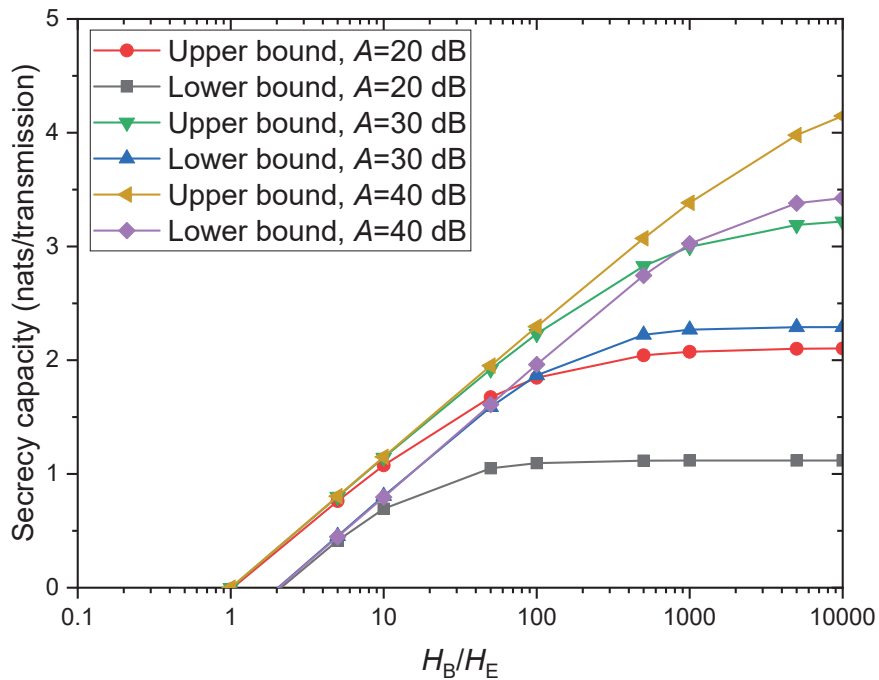}}
\caption{Secrecy capacity bounds versus ${H_{\rm{B}}}/{H_{\rm{E}}}$ with different $A$ when $\xi  = 0.3$ and $A/P = 1.5$}
\label{fig8}
\end{figure}

Fig. \ref{fig9} plots the secrecy capacity bounds versus $A$ for different VLC scenarios when $\xi  = 0.3$, $A/P = 1$, and ${H_{\rm{B}}}/{H_{\rm{E}}} = 1000$.
Here, the secrecy capacity bounds for VLC having the signal-dependent noise (i.e., lower bound (\ref{eq26}) and upper bound (\ref{eq31}) in this paper),
the secrecy capacity bounds for VLC having the signal-independent noise (i.e., lower bound (20) and upper bound (26) in \cite{BIB28}),
and the channel capacity bounds for VLC having the signal-dependent noise (i.e., lower bound (17) and upper bound (25) in \cite{BIB16}),
the channel capacity bounds for VLC having the signal-independent noise (i.e., lower bound (17) and upper bound (34) in \cite{BIB06}) are provided.
As can be observed, the trends of secrecy capacity bounds and channel capacity bounds are different.
Specifically, no matter the noise is signal-dependent or signal-independent, the channel capacity bounds in \cite{BIB06} and \cite{BIB16} increase approximately linearly with $A$.
However, as the increase of $A$, the secrecy capacity bounds in this paper and \cite{BIB28} increase first and then tend to stable values.
Moreover, the signal-dependent noise degrades the channel capacity or the secrecy capacity of the VLC system.
Furthermore, at large $A$, for the VLC system having the signal-independent noise or signal-dependent noise, the channel capacity bounds are always larger than the secrecy capacity, this is because the system performance degrades due to the wiretap of Eve.

\begin{figure}
\centering
\rotatebox[origin=c]{-90}{\includegraphics[width=9cm]{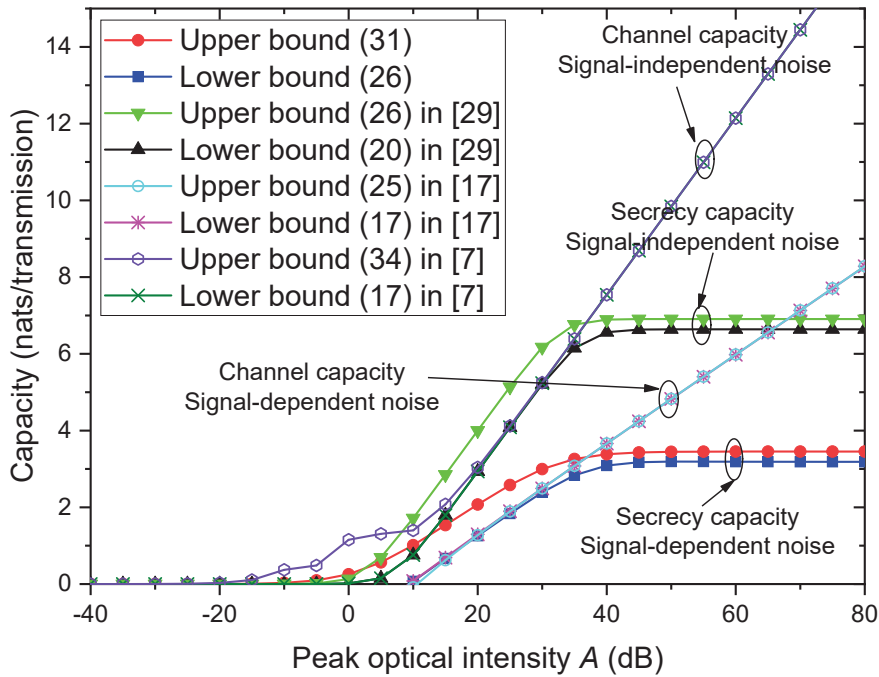}}
\caption{Capacity bounds versus peak optical intensity $A$ for different scenarios when $\xi  = 0.3$, $A/P = 1$ and ${H_{\rm{B}}}/{H_{\rm{E}}} =1000$}
\label{fig9}
\end{figure}

\section{Conclusions}
\label{section6}
We analyzed the PLS performance for the indoor VLC system having the signal-dependent noise.
Two kinds of signal constraints were considered:
one had the non-negativity and average optical intensity constraints;
the other had the non-negativity, average optical intensity and peak optical intensity constraints.
By considering these two kinds of signal constraints,
we derived tight lower and upper bounds on the secrecy capacity.
Letting the signal-dependent noise variance approach zero,
we derived the secrecy capacity bounds for the system having only the signal-independent noise.
For large optical intensity,
we also derived tight asymptotic secrecy capacity bounds.
From numerical results, we observed that the signal-dependent Gaussian noise has a strong impact on secrecy capacity,
thereby degrading the PLS performance.
At large optical intensity, we observed that the secrecy capacity does not increase with $P$ or $A$, but keeps stable value. Such an observation for secrecy capacity is different from that for channel capacity.
After deriving the secrecy capacity bounds, our future work will focus on exploiting novel schemes to improve the PLS performance of the VLC system. Moreover, experimental platforms will also be established in future work to validate the derived results.

\numberwithin{equation}{section}
\appendices
\section{Proof of \emph{Theorem \ref{the1}}}
\label{Appendix_A}
In (\ref{eq13}), we can obtain a lower bound on secrecy capacity by maximizing the source entropy under the constraints (\ref{eq2}) and (\ref{eq4}), i.e.,
\begin{eqnarray}
&&\mathop {\max }\limits_{{f_X}\left( x \right)} {\cal H}( X) =  - \int_0^\infty  {{f_X} (x){\rm{ln}}\left[ {{f_X}( x )} \right]{\rm{d}}x} \nonumber \\
{\rm{s}}{\rm{.t}}{\rm{.}}\;&&\int_0^\infty  {{f_X}( x){\rm{d}}x}  = 1 \label{eq35}\\
&&{E_X}\left( X \right) = \int_0^\infty  {x{f_X}( x){\rm{d}}x}  = \xi P. \nonumber
\end{eqnarray}
By using the variational method, we obtain an optimized PDF of the input $X$ as
\begin{eqnarray}
{f_X}(x) = \left\{ \begin{array}{l}
\frac{1}{{\xi P}}{e^{ - \frac{1}{{\xi P}}x}},\;{\rm{if}}\,x > 0\\
0,\quad \quad \quad {\rm{otherwise}}
\end{array} \right..
\label{eq36}
\end{eqnarray}
According to the input PDF (\ref{eq36}), ${\cal H}( X)$ is given by
\begin{equation}
{\cal H}( X ) = \ln \left( {e\xi P} \right).
\label{eq37}
\end{equation}
Moreover, for the negative exponential distribution (\ref{eq36}), we can obtain
\begin{eqnarray}
\left\{ \begin{array}{l}
{E_X}(X) = \xi P\\
{\rm var}(X) = {\xi ^2}{P^2}\\
{E_{X,{Z_{{\rm{E}},1}}}}(\sqrt X {Z_{{\rm{E}},1}}) = {E_X}(\sqrt X ){E_{{Z_{{\rm{E}},1}}}}({Z_{{\rm{E}},1}}) = 0
\end{array} \right..
\label{eq38}
\end{eqnarray}
The last equality in (\ref{eq38}) holds because $X$ and ${Z_{{\rm{E}},1}}$ are independent variables and the mean of ${Z_{{\rm{E}},1}}$ is zero.

Then, we can obtain ${\mathop{\rm var}} \left( {{Y_{\rm{E}}}} \right)$ as
\begin{eqnarray}
{\mathop{\rm var}} \left( {{Y_{\rm{E}}}} \right) \!\!\!&=&\!\! {\rm{var}}\left( {{H_{\rm{E}}}X{\rm{ + }}\sqrt {{H_{\rm{E}}}X} {Z_{{\rm{E}},1}}} \right) + {\mathop{\rm var}} \left( {{Z_{{\rm{E,0}}}}} \right) \nonumber \\
 &=&\!\! {\rm{var}}\left( {{H_{\rm{E}}}X} \right){\rm{ + var}}\left( {\sqrt {{H_{\rm{E}}}X} {Z_{{\rm{E}},1}}} \right) \nonumber\\
 &+&\!\! {\rm{2}}{E_{_{X,{Z_{{\rm{E}},1}}}}}\left\{ {\left[ {{H_{\rm{E}}}X - E_{X} \left( {{H_{\rm{E}}}X} \right)} \right]\left[ {\sqrt {{H_{\rm{E}}}X} {Z_{{\rm{E}},1}} - E_{X,{Z_{{\rm{E}},1}}} (\sqrt {{H_{\rm{E}}}X} {Z_{{\rm{E}},1}})} \right]} \right\} + \sigma _{\rm{E}}^2 \nonumber\\
 &=&\!\! H_{\rm{E}}^2{\xi ^2}{P^2} + {H_{\rm{E}}}{\mathop{\rm var}} \left( {\sqrt X {Z_{{\rm{E}},1}}} \right) + \sigma _{\rm{E}}^2.
\label{eq39}
\end{eqnarray}
Moreover, ${\mathop{\rm var}} ( {\sqrt X {Z_{{\rm{E}},1}}} )$ can be further expressed as
\begin{eqnarray}
{\mathop{\rm var}} \left( {\sqrt X {Z_{{\rm{E}},1}}} \right) &=& {E_{_{X,{Z_{{\rm{E}},1}}}}}(XZ_{{\rm{E}},1}^2) - {\left[ {{E_{_{X,{Z_{{\rm{E}},1}}}}}\left( {\sqrt X {Z_{{\rm{E}},1}}} \right)} \right]^2} \nonumber\\
 &=& {E_X}(X){E_{Z_{{\rm{E}},1}^{}}}(Z_{{\rm{E}},1}^2) - 0 \nonumber\\
 &=& \xi P\varsigma _{\rm{E}}^2\sigma _{\rm{E}}^2.
\label{eq40}
\end{eqnarray}
Substituting (\ref{eq40}) into (\ref{eq39}), we can finally obtain the variance of $Y_{\rm{E}}$ as
\begin{equation}
{\mathop{\rm var}} \left( {{Y_{\rm{E}}}} \right){\rm{ = }}H_{\rm{E}}^2{\xi ^2}{P^2} + {H_{\rm{E}}}\xi P\varsigma _{\rm{E}}^2\sigma _{\rm{E}}^2 + \sigma _{\rm{E}}^2.
\label{eq41}
\end{equation}
Moreover, ${E_X}\left[ {\ln \left( {\frac{{1 + {H_{\rm{E}}}\varsigma _{\rm{E}}^2X}}{{1 + {H_{\rm{B}}}\varsigma _{\rm{B}}^2X}}} \right)} \right]$ can be obtained by using the optimized PDF, i.e.,
\begin{eqnarray}
{E_X}\left[ {\ln \left( {\frac{{1 + {H_{\rm{E}}}\varsigma _{\rm{E}}^2X}}{{1 + {H_{\rm{B}}}\varsigma _{\rm{B}}^2X}}} \right)} \right] \!\!\!\!\!
 &=&\!\!\!\!  - \int_0^\infty  {\ln \left( {\frac{{1 + {H_{\rm{E}}}\varsigma _{\rm{E}}^2X}}{{1 + {H_{\rm{B}}}\varsigma _{\rm{B}}^2X}}} \right){\rm{d}}\left( {{e^{ - \frac{1}{{\xi P}}x}}} \right)} \nonumber\\
 &=&\!\!\!\! \left[\! {{e^{\frac{1}{{{H_{\rm{B}}}\varsigma _{\rm{B}}^2\xi P}}}}Ei\left(\! { - \frac{1}{{{H_{\rm{B}}}\varsigma _{\rm{B}}^2\xi P}}} \!\right) \!-\! {e^{\frac{1}{{{H_{\rm{E}}}\varsigma _{\rm{E}}^2\xi P}}}}Ei\left(\! { - \frac{1}{{{H_{\rm{E}}}\varsigma _{\rm{E}}^2\xi P}}} \!\right)} \!\right].
\label{eq42}
\end{eqnarray}
Substituting (\ref{eq37}), (\ref{eq39}) and (\ref{eq42}) into the lower bound (\ref{eq13}), we obtain \emph{Theorem \ref{the1}}. \;\;\;\;\;\;\;\;\;\;\;$\blacksquare$

\section{Proof of \emph{Theorem \ref{the2}}}
\label{Appendix_B}
According to the secrecy capacity's dual expression (\ref{eq19}), the secrecy capacity is upper-bounded as
\begin{equation}
{C_s} \le {E_{{X^*}{Y_{\rm{E}}}}}\left\{ {D\left( {\left. {{f_{{Y_{\rm{B}}}|X{Y_{\rm E}}}}({y_{\rm{B}}}|X,{Y_{\rm{E}}})} \right\|{g_{{Y_{\rm{B}}}|{Y_{\rm E}}}}({y_{\rm{B}}}|{Y_{\rm{E}}})} \right)} \right\}.
\label{eq43}
\end{equation}

Using the concept of relative entropy, eq. (\ref{eq43}) can be decomposed into two parts
\begin{eqnarray}
{C_s} &\le& {E_{{X^*}{Y_{\rm{E}}}}}\left\{ {\int_{ - \infty }^\infty  {{f_{{Y_{\rm{B}}}|X{Y_{\rm E}}}}({y_{\rm{B}}}|X,{Y_{\rm{E}}})\ln \left[ {\frac{{{f_{{Y_{\rm{B}}}|X{Y_{\rm E}}}}({y_{\rm{B}}}|X,{Y_{\rm{E}}})}}{{{g_{{Y_{\rm{B}}}|{Y_{\rm E}}}}({y_{\rm{B}}}|{Y_{\rm{E}}})}}} \right]{\rm{d}}{y_{\rm{B}}}} } \right\} \nonumber\\
 &=& \underbrace {{E_{{X^*}{Y_{\rm{E}}}}}\left\{ {\int_{ - \infty }^\infty  {{f_{{Y_{\rm{B}}}|X{Y_{\rm E}}}}({y_{\rm{B}}}|X,{Y_{\rm{E}}})\ln \left[ {{f_{{Y_{\rm{B}}}|X{Y_{\rm E}}}}({y_{\rm{B}}}|X,{Y_{\rm{E}}})} \right]{\rm{d}}{y_{\rm{B}}}} } \right\}}_{{I_1}} \nonumber\\
&&\underbrace { - {E_{{X^*}{Y_{\rm{E}}}}}\left\{ {\int_{ - \infty }^\infty  {{f_{{Y_{\rm{B}}}|X{Y_{\rm E}}}}({y_{\rm{B}}}|X,{Y_{\rm{E}}})\ln \left[ {{g_{{Y_{\rm{B}}}|{Y_{\rm E}}}}({y_{\rm{B}}}|{Y_{\rm{E}}})} \right]{\rm{d}}{y_{\rm{B}}}} } \right\}}_{{I_2}},
\label{eq44}
\end{eqnarray}
where ${I_1}$ in (\ref{eq44}) can be written as
\begin{eqnarray}
{I_1} &=&  - {\cal H}({Y_{\rm{B}}}|{X^*},{Y_{\rm{E}}}) \nonumber\\
 &=&  - \left[ {{\cal H}({Y_{\rm{B}}}|{X^*}) + {\cal H}({Y_{\rm{E}}}|{X^*},{Y_{\rm{B}}}) - {\cal H}({Y_{\rm{E}}}|{X^*})} \right],
\label{eq45}
\end{eqnarray}
where ${\cal H}\left( {{Y_{\rm{B}}}\left| {{X^*}} \right.} \right)$ is given by
\begin{eqnarray}
{\cal H}\left( {{Y_{\rm{B}}}\left| {{X^*}} \right.} \right) \!\!\!&=&\!\!\!  - {E_{{X^*}}}\left\{ {\int_{ - \infty }^\infty  {\frac{{{e^{ - \frac{{{{({y_{\rm{B}}} - {H_{\rm{B}}}X)}^2}}}{{2\left( {1 + {H_{\rm{B}}}X\varsigma _{\rm{B}}^2} \right)\sigma _{\rm{B}}^2}}}}}}{{\sqrt {2\pi \left( {1 + {H_{\rm{B}}}X\varsigma _{\rm{B}}^2} \right)\sigma _{\rm{B}}^2} }}} \ln \left[ {\frac{{{e^{ - \frac{{{{({y_{\rm{B}}} - {H_{\rm{B}}}X)}^2}}}{{2\left( {1 + {H_{\rm{B}}}X\varsigma _{\rm{B}}^2} \right)\sigma _{\rm{B}}^2}}}}}}{{\sqrt {2\pi \left( {1 + {H_{\rm{B}}}X\varsigma _{\rm{B}}^2} \right)\sigma _{\rm{B}}^2} }}} \right]{\rm{d}}{y_{\rm{B}}}} \right\} \nonumber\\
 &=&\!\!\! {E_{{X^*}}}\left\{ {\frac{1}{2}\ln \left[ {2\pi e\left( {1 + {H_{\rm{B}}}X\varsigma _{\rm{B}}^2} \right)\sigma _{\rm{B}}^2} \right]} \right\}.
\label{eq46}
\end{eqnarray}
Similarly, ${\cal H}\left( {{Y_{\rm{E}}}\left| {{X^*}} \right.} \right)$ is given by
\begin{equation}
{\cal H}\left( {{Y_{\rm{E}}}\left| {{X^*}} \right.} \right) = {E_{{X^*}}}\left\{ {\frac{1}{2}\ln \left[ {2\pi e\left( {1 + {H_{\rm{E}}}X\varsigma _{\rm{E}}^2} \right)\sigma _{\rm{E}}^2} \right]} \right\}.
\label{eq47}
\end{equation}

According to the received signal model (\ref{eq1}), we can rewrite ${Y_{\rm{E}}}$ as
\begin{equation}
{Y_{\rm{E}}} = \frac{{{H_{\rm{E}}}}}{{{H_{\rm{B}}}}}{Y_{\rm{B}}} - \frac{{{H_{\rm{E}}}}}{{{H_{\rm{B}}}}}\sqrt {{H_{\rm{B}}}X} {Z_{{\rm{B}},1}} - \frac{{{H_{\rm{E}}}}}{{{H_{\rm{B}}}}}{Z_{{\rm{B}},0}} + \sqrt {{H_{\rm{E}}}X} {Z_{{\rm{E}},1}} + {Z_{{\rm{E}},0}}.
\label{eq48}
\end{equation}

The conditional PDF ${f_{{Y_{\rm{E}}}|X{Y_{\rm{B}}}}}({y_{\rm{E}}}|X,{Y_{\rm{B}}})$ is given by
\begin{equation}
{f_{{Y_{\rm{E}}}|X{Y_{\rm{B}}}}}({y_{\rm{E}}}|X,{Y_{\rm{B}}}) = \frac{{{e^{ - \frac{{{{\left( {{y_{\rm{E}}} - \frac{{{H_{\rm{E}}}}}{{{H_{\rm{B}}}}}{Y_{\rm{B}}}} \right)}^2}}}{{2\left( {MX + N} \right)}}}}}}{{\sqrt {2\pi \left( {MX + N} \right)} }},
\label{eq49}
\end{equation}
where $M = H_{\rm{E}}^2\varsigma _{\rm{B}}^2\sigma _{\rm{B}}^2/{H_{\rm{B}}} + {H_{\rm{E}}}\varsigma _{\rm{E}}^2\sigma _{\rm{E}}^2$ and $N = H_{\rm{E}}^2\sigma _{\rm{B}}^2/H_{\rm{B}}^2 + \sigma _{\rm{E}}^2$.
According the conditional PDF (\ref{eq49}), we can derive ${\cal H}({Y_{\rm{E}}}|{X^*},{Y_{\rm{B}}})$ as
\begin{equation}
{\cal H}({Y_{\rm{E}}}|{X^*},{Y_{\rm{B}}}) = {E_{{X^*}}}\left\{ {\frac{1}{2}\ln \left[ {2\pi e\left( {MX + N} \right)} \right]} \right\}.
\label{eq50}
\end{equation}

Substituting (\ref{eq46}), (\ref{eq47}) and (\ref{eq50}) into (\ref{eq45}), ${I_1}$ can be written as
\begin{eqnarray}
{I_1} &=&  - {E_{{X^*}}}\left\{ {\frac{1}{2}\ln \left[ {\frac{{2\pi e\sigma _{\rm{B}}^2\left( {MX + N} \right)\left( {1 + {H_{\rm{B}}}X\varsigma _{\rm{B}}^2} \right)}}{{\sigma _{\rm{E}}^2\left( {1 + {H_{\rm{E}}}X\varsigma _{\rm{E}}^2} \right)}}} \right]} \right\} \nonumber\\
 &=&  - \frac{1}{2}\ln \left( {2\pi e\frac{{\sigma _{\rm{B}}^2}}{{\sigma _{\rm{E}}^2}}} \right) - \frac{1}{2}{E_{{X^*}}}\left\{ {\ln \left( {MX + N} \right) + \ln \left( {\frac{{1 + {H_{\rm{B}}}X\varsigma _{\rm{B}}^2}}{{1 + {H_{\rm{E}}}X\varsigma _{\rm{E}}^2}}} \right)} \right\}.
\label{eq51}
\end{eqnarray}

To obtain ${I_2}$ in (\ref{eq44}), we select ${g_{{Y_{\rm{B}}}|{Y_{\rm E}}}}({y_{\rm{B}}}|{y_{\rm{E}}})$ as
\begin{equation}
{g_{{Y_{\rm{B}}}|{Y_{\rm E}}}}({y_{\rm{B}}}|{y_{\rm{E}}}) = \frac{1}{{2{s^2}}}{e^{ - \frac{{\left| {{y_{\rm{B}}} - \mu {y_{\rm{E}}}} \right|}}{{{s^2}}}}},
\label{eq52}
\end{equation}
where $\mu $ and $s$ are two free parameters to be determined.

Substituting the selected conditional PDF (\ref{eq52}) into (\ref{eq44}), we can obtain ${I_2}$ as
\begin{eqnarray}
{I_2} &=&  - {E_{{X^*}{Y_{\rm{B}}}}}\left\{ {\int_{ - \infty }^\infty  {{f_{{Y_{\rm{E}}}|X{Y_{\rm{B}}}}}({y_{\rm{E}}}|X,{Y_{\rm{B}}})\ln \left[ {\frac{1}{{2{s^2}}}{e^{ - \frac{{\left| {{Y_{\rm{B}}} - \mu {y_{\rm{E}}}} \right|}}{{{s^2}}}}}} \right]{\rm{d}}{y_{\rm{E}}}} } \right\} \nonumber\\
 &=& {E_{{X^*}{Y_{\rm{B}}}}}\left\{ {\ln \left( {2{s^2}} \right) + \frac{1}{{{s^2}}}\underbrace {\int_{ - \infty }^\infty  {\frac{{{e^{ - \frac{{{{\left( {{y_{\rm{E}}} - \frac{{{H_{\rm{E}}}}}{{{H_{\rm{B}}}}}{Y_{\rm{B}}}} \right)}^2}}}{{2\left( {MX + N} \right)}}}}}}{{\sqrt {2\pi \left( {MX + N} \right)} }}\left| {{Y_{\rm{B}}} - \mu {y_{\rm{E}}}} \right|{\rm{d}}{y_{\rm{E}}}} }_{{C_1}}} \right\}.
\label{eq53}
\end{eqnarray}
Let $t = {y_{\rm{E}}} - {H_{\rm{E}}}{Y_{\rm{B}}}/{H_{\rm{B}}}$, ${C_1}$ in (\ref{eq53}) can be expressed as
\begin{eqnarray}
{C_1} &=& \int_{ - \infty }^\infty  {\frac{{{e^{ - \frac{{{t^2}}}{{2\left( {MX + N} \right)}}}}}}{{\sqrt {2\pi \left( {MX + N} \right)} }}\left| {\left( {1 - \mu \frac{{{H_{\rm{E}}}}}{{{H_{\rm{B}}}}}} \right){Y_{\rm{B}}} - \mu t} \right|{\rm{d}}t} \nonumber\\
 &\le& \int_{ - \infty }^\infty  {\frac{{{e^{ - \frac{{{t^2}}}{{2\left( {MX + N} \right)}}}}}}{{\sqrt {2\pi \left( {MX + N} \right)} }}\left( {\left| {\mu t} \right| + \left| {\left( {1 - \mu \frac{{{H_{\rm{E}}}}}{{{H_{\rm{B}}}}}} \right){Y_{\rm{B}}}} \right|} \right){\rm{d}}t} \nonumber\\
 &=& \int_{ - \infty }^\infty  {\frac{{{e^{ - \frac{{{t^2}}}{{2\left( {MX + N} \right)}}}}}}{{\sqrt {2\pi \left( {MX + N} \right)} }}\left( {\left| \mu  \right|\left| t \right| + \left| {1 - \mu \frac{{{H_{\rm{E}}}}}{{{H_{\rm{B}}}}}} \right|\left| {{Y_{\rm{B}}}} \right|} \right){\rm{d}}t} \nonumber\\
 &=& 2\left| \mu  \right|\sqrt {\frac{{MX + N}}{{2\pi }}}  + \left| {1 - \mu \frac{{{H_{\rm{E}}}}}{{{H_{\rm{B}}}}}} \right|\left| {{Y_{\rm{B}}}} \right|.
\label{eq54}
\end{eqnarray}
The inequality in (\ref{eq54}) holds due to the fact that $\left| {a - b} \right| \le \left| a \right| + \left| b \right|$.

Substituting (\ref{eq54}) into (\ref{eq53}), we have
\begin{eqnarray}
{I_2} \!\!\!\!\!&=&\!\!\!\!\! {E_{{X^*}{Y_{\rm{B}}}}}\left\{ {\ln \left( {2{s^2}} \right) + \frac{1}{{{s^2}}}\left[ {2\left| \mu  \right|\sqrt {\frac{{MX + N}}{{2\pi }}}  + \left| {1 - \mu \frac{{{H_{\rm{E}}}}}{{{H_{\rm{B}}}}}} \right|\left| {{Y_{\rm{B}}}} \right|} \right]} \right\} \nonumber\\
 &=&\!\!\!\!\! {E_{{X^*}}}\!\!\left\{\!\! {\ln \!\left( {2{s^2}} \right) \!\!+\!\! \frac{1}{{{s^2}}}\!\!\left[\! {2\!\left| \mu  \right|\!\sqrt {\frac{{MX \!+\! N}}{{2\pi }}}  \!\!+\!\! \left|\! {1 \!-\! \mu \frac{{{H_{\rm{E}}}}}{{{H_{\rm{B}}}}}}\! \right|\!\!\underbrace {\int_{ - \infty }^\infty \!\! {\frac{{{e^{ - \frac{{{{\left( {{y_{\rm{B}}} - {H_{\rm{B}}}X} \right)}^2}}}{{2\left( {1 + {H_{\rm{B}}}X\varsigma _{\rm{B}}^2} \right)\sigma _{\rm{B}}^2}}}}}}{{\sqrt {2\pi \left( {1 \!\!+\!\! {H_{\rm{B}}}X\varsigma _{\rm{B}}^2} \right)\sigma _{\rm{B}}^2} }}\!\left| {{y_{\rm{B}}}} \right|\!{\rm{d}}{y_{\rm{B}}}} }_{{C_2}}} \!\!\right]} \!\!\!\right\}\!\!,
\label{eq55}
\end{eqnarray}
where ${C_2}$ can be expressed as
\begin{eqnarray}
{C_2} &=& \int_{ - \infty }^\infty  {\frac{{{e^{ - \frac{{{h^2}}}{{2\left( {1 + {H_{\rm{B}}}X\varsigma _{\rm{B}}^2} \right)\sigma _{\rm{B}}^2}}}}}}{{\sqrt {2\pi \left( {1 + {H_{\rm{B}}}X\varsigma _{\rm{B}}^2} \right)\sigma _{\rm{B}}^2} }}\left| {h + {H_{\rm{B}}}X} \right|{\rm{d}}h} \nonumber\\
 &\le& \int_{ - \infty }^\infty  {\frac{{{e^{ - \frac{{{h^2}}}{{2\left( {1 + {H_{\rm{B}}}X\varsigma _{\rm{B}}^2} \right)\sigma _{\rm{B}}^2}}}}}}{{\sqrt {2\pi \left( {1 + {H_{\rm{B}}}X\varsigma _{\rm{B}}^2} \right)\sigma _{\rm{B}}^2} }}\left( {\left| h \right| + \left| {{H_{\rm{B}}}X} \right|} \right){\rm{d}}h} \nonumber\\
 &=& 2\sqrt {\frac{{\left( {1 + {H_{\rm{B}}}X\varsigma _{\rm{B}}^2} \right)\sigma _{\rm{B}}^2}}{{2\pi }}}  + {H_{\rm{B}}}X.
\label{eq56}
\end{eqnarray}
The inequality in (\ref{eq56}) holds because $\left| {a + b} \right| \le \left| a \right| + \left| b \right|$.

To obtain a tight upper bound, we perform the first partial derivative of (\ref{eq55}) with ${s^2}$ to get a minimum upper bound of ${I_2}$. Therefore, the minimal point is derived as
\begin{equation}
{s^2} = 2\left| \mu  \right|\sqrt {\frac{{MX + N}}{{2\pi }}}  + \left| {1 - \mu \frac{{{H_{\rm{E}}}}}{{{H_{\rm{B}}}}}} \right|{C_2}.
\label{eq57}
\end{equation}
Then, ${I_2}$ is upper-bounded by
\begin{equation}
{I_2} \le {E_{{X^*}}}\left\{\! {\ln\!\! \left[ {4e\left( {\left| \mu  \right|\sqrt {\frac{{MX + N}}{{2\pi }}}  \!\!+\!\! \left| {1 - \mu \frac{{{H_{\rm{E}}}}}{{{H_{\rm{B}}}}}} \right|\left( {\sqrt {\frac{{\left( {1 + {H_{\rm{B}}}X\varsigma _{\rm{B}}^2} \right)\sigma _{\rm{B}}^2}}{{2\pi }}}  + \frac{{{H_{\rm{B}}}X}}{2}} \!\right)} \!\right)} \!\right]} \!\right\}.
\label{eq58}
\end{equation}

Substituting (\ref{eq51}) and (\ref{eq58}) into (\ref{eq44}), ${C_s}$ is upper-bounded by
\begin{eqnarray}
{C_s} \!\!\!\!&\le&\!\!\!\!  - \frac{1}{2}\ln \left( {2\pi e\frac{{\sigma _{\rm{B}}^2}}{{\sigma _{\rm{E}}^2}}} \right) - \frac{1}{2}{E_{{X^*}}}\left\{ {\ln \left( {\frac{{1 + {H_{\rm{B}}}X\varsigma _{\rm{B}}^2}}{{1 + {H_{\rm{E}}}X\varsigma _{\rm{E}}^2}}} \right)} \right\} \nonumber\\
 &&\!\!\!\!\!+ {E_{{X^*}}}\!\left\{\!\! {\ln \!\!\left[ {4e\!\left( {\frac{{\left| \mu  \right|}}{{\sqrt {2\pi } }} \!+\! \left| {1 \!-\! \mu \frac{{{H_{\rm{E}}}}}{{{H_{\rm{B}}}}}} \right|\!\!\left( {\sqrt {\frac{{\left( {1 + {H_{\rm{B}}}X\varsigma _{\rm{B}}^2} \right)\sigma _{\rm{B}}^2}}{{2\pi (MX + N)}}}  \!+\! \frac{{{H_{\rm{B}}}X}}{{2\sqrt {MX + N} }}} \!\right)} \!\right)} \!\right]} \!\right\}.
\label{eq59}
\end{eqnarray}
Since $\frac{{{H_{\rm{B}}}X}}{{2\sqrt {MX + N} }} \le \frac{{{H_{\rm{B}}}X}}{{2\sqrt {MX} }} = \frac{{{H_{\rm{B}}}\sqrt X }}{{2\sqrt M }}$, eq. (\ref{eq59}) can be further written as
\begin{eqnarray}
{C_s} \!\!\!&\le&\!\!\!  - \frac{1}{2}\ln \left( {2\pi e\frac{{\sigma _{\rm{B}}^2}}{{\sigma _{\rm{E}}^2}}} \right) - \frac{1}{2}{E_{{X^*}}}\left\{ {\ln \left( {\frac{{1 + {H_{\rm{B}}}X\varsigma _{\rm{B}}^2}}{{1 + {H_{\rm{E}}}X\varsigma _{\rm{E}}^2}}} \right)} \right\} \nonumber\\
 &&\!\!\!\!+ {E_{{X^*}}}\left\{ {\ln \left[ {4e\left( {\frac{{\left| \mu  \right|}}{{\sqrt {2\pi } }} + \left| {1 - \mu \frac{{{H_{\rm{E}}}}}{{{H_{\rm{B}}}}}} \right|\left( {\sqrt {\frac{{\left( {1 + {H_{\rm{B}}}X\varsigma _{\rm{B}}^2} \right)\sigma _{\rm{B}}^2}}{{2\pi (MX + N)}}}  + \frac{{{H_{\rm{B}}}\sqrt X }}{{2\sqrt M }}} \right)} \right)} \right]} \right\}.
\label{eq60}
\end{eqnarray}
Then, by using Jenson's inequality for the convex function $\ln ( \cdot )$ and $\sqrt {\, \cdot \,} $, we have
\begin{eqnarray}
{C_s} \!\!\!\!&\le&\!\!\!\!  - \frac{1}{2}\ln \left( {2\pi e\frac{{\sigma _{\rm{B}}^2}}{{\sigma _{\rm{E}}^2}}} \right) - \frac{1}{2}\ln \left[ {{E_{{X^*}}}\left( {\frac{{1 + {H_{\rm{B}}}X\varsigma _{\rm{B}}^2}}{{1 + {H_{\rm{E}}}X\varsigma _{\rm{E}}^2}}} \right)} \right] \nonumber\\
 &&\!\!\!\!\!+ \ln \left[ {4e\left( {\frac{{\left| \mu  \right|}}{{\sqrt {2\pi } }} + \left| {1 - \mu \frac{{{H_{\rm{E}}}}}{{{H_{\rm{B}}}}}} \right|\left( {{E_{{X^*}}}\left( {\sqrt {\frac{{\left( {1 + {H_{\rm{B}}}X\varsigma _{\rm{B}}^2} \right)\sigma _{\rm{B}}^2}}{{2\pi (MX + N)}}} } \right) + \frac{{{H_{\rm{B}}}}}{{2\sqrt M }}{E_{{X^*}}}\left( {\sqrt X } \!\right)} \!\right)} \!\right)} \!\right] \nonumber\\
 &\le&\!\!\!\!  - \frac{1}{2}\ln \left( {2\pi e\frac{{\sigma _{\rm{B}}^2}}{{\sigma _{\rm{E}}^2}}} \right) - \frac{1}{2}\ln \left[ {{E_{{X^*}}}\left( {\frac{{1 + {H_{\rm{B}}}X\varsigma _{\rm{B}}^2}}{{1 + {H_{\rm{E}}}X\varsigma _{\rm{E}}^2}}} \right)} \right] \nonumber\\
 &&\!\!\!\!\!+ \ln \left[ {4e\left( {\frac{{\left| \mu  \right|}}{{\sqrt {2\pi } }} + \left| {1 - \mu \frac{{{H_{\rm{E}}}}}{{{H_{\rm{B}}}}}} \right|\left( {\sqrt {{E_{{X^*}}}\left[ {\frac{{\left( {1 + {H_{\rm{B}}}X\varsigma _{\rm{B}}^2} \right)\sigma _{\rm{B}}^2}}{{2\pi (MX + N)}}} \right]}  + \frac{{{H_{\rm{B}}}}}{{2\sqrt M }}\sqrt {{E_{{X^*}}}(X)} } \!\right)} \!\right)} \!\right] \nonumber\\
 &=&  - \frac{1}{2}\ln \left( {2\pi e\frac{{\sigma _{\rm{B}}^2}}{{\sigma _{\rm{E}}^2}}} \right) - \frac{1}{2}\ln \left[ {{E_{{X^*}}}\left( {\frac{{1 + {H_{\rm{B}}}X\varsigma _{\rm{B}}^2}}{{1 + {H_{\rm{E}}}X\varsigma _{\rm{E}}^2}}} \right)} \right] \nonumber\\
 &&+ \ln \left[ {4e\left( {\frac{{\left| \mu  \right|}}{{\sqrt {2\pi } }} + \left| {1 - \mu \frac{{{H_{\rm{E}}}}}{{{H_{\rm{B}}}}}} \right|\left( {\sqrt {{E_{{X^*}}}\left[ {\frac{{\left( {1 + {H_{\rm{B}}}X\varsigma _{\rm{B}}^2} \right)\sigma _{\rm{B}}^2}}{{2\pi (MX + N)}}} \right]}  + \frac{{{H_{\rm{B}}}\sqrt {\xi P} }}{{2\sqrt M }}} \right)} \right)} \right].
\label{eq61}
\end{eqnarray}
In VLC, the received optical intensity is often large to satisfy the illumination requirements.
Therefore, we are more interested in the performance at large optical intensity.
As the nominal optical intensity $P$ is loosened to infinity,
the secrecy-capacity-achieving input distribution ${f_{{X^*}}}(x)$ will escape to infinity \cite{BIB16}, \cite{BIB34}.
This indicates that when $P$ approaches infinity, the probability for any set of finite-intensity input symbols will tend to zero.
Therefore, when $P$ tends to infinity, we have
\begin{eqnarray}
\left\{ \begin{array}{l}
{E_{{X^*}}}\left( {\frac{{1 + {H_{\rm{B}}}X\varsigma _{\rm{B}}^2}}{{1 + {H_{\rm{E}}}X\varsigma _{\rm{E}}^2}}} \right) =  \frac{{1 + {H_{\rm{B}}}\varsigma _{\rm{B}}^2P}}{{1 + {H_{\rm{E}}}\varsigma _{\rm{E}}^2P}}\\
{E_{{X^*}}}\left[ {\frac{{\left( {1 + {H_{\rm{B}}}X\varsigma _{\rm{B}}^2} \right)\sigma _{\rm{B}}^2}}{{2\pi (MX + N)}}} \right] =  \frac{{\left( {1 + {H_{\rm{B}}}P\varsigma _{\rm{B}}^2} \right)\sigma _{\rm{B}}^2}}{{2\pi (MP + N)}}
\end{array} \right..
\label{eq62}
\end{eqnarray}
Moreover, by using L'Hospital's rule for (\ref{eq62}), we have
\begin{eqnarray}
\left\{ \begin{array}{l}
\mathop {\lim }\limits_{P \to \infty } \frac{{1 + {H_{\rm{B}}}\varsigma _{\rm{B}}^2P}}{{1 + {H_{\rm{B}}}\varsigma _{\rm{B}}^2P}} = \frac{{{H_{\rm{B}}}\varsigma _{\rm{B}}^2}}{{{H_{\rm{E}}}\varsigma _{\rm{E}}^2}}\\
\mathop {\lim }\limits_{P \to \infty } \frac{{\left( {1 + {H_{\rm{B}}}P\varsigma _{\rm{B}}^2} \right)\sigma _{\rm{B}}^2}}{{2\pi (MP + N)}} = \frac{{{H_{\rm{B}}}\varsigma _{\rm{B}}^2\sigma _{\rm{B}}^2}}{{2\pi M}}
\end{array} \right..
\label{eq63}
\end{eqnarray}
Substituting (\ref{eq62}) and (\ref{eq63}) into (\ref{eq61}), we have
\begin{equation}
{C_s} \le  - \frac{1}{2}\ln \left( {2\pi e\frac{{{H_{\rm{B}}}\varsigma _{\rm{B}}^2\sigma _{\rm{B}}^2}}{{{H_{\rm{E}}}\varsigma _{\rm{E}}^2\sigma _{\rm{E}}^2}}} \right) \!+\! \ln \!\!\left[\! {4e\!\!\left(\! {\underbrace {\frac{{\left| \mu  \right|}}{{\sqrt {2\pi } }} \!+\! \left| {1 \!-\! \mu \frac{{{H_{\rm{E}}}}}{{{H_{\rm{B}}}}}} \right|\!\left( {\sqrt {\frac{{{H_{\rm{B}}}\varsigma _{\rm{B}}^2\sigma _{\rm{B}}^2}}{{2\pi M}}}  \!+\! \frac{{{H_{\rm{B}}}}}{2}\sqrt {\frac{{\xi P}}{M}} } \right)}_{{C_{\rm{3}}}}} \!\!\right)} \!\right].
\label{eq64}
\end{equation}
In (\ref{eq64}), ${C_{\rm{3}}}$ is a function with respect to $\mu $.
To obtain a tight upper bound on secrecy capacity, we should determine a smaller value of ${C_{\rm{3}}}$.
In the following, three cases are considered:

\textbf{Case 1:} when $\mu  \le 0$, ${C_{\rm{3}}}$ is given by
\begin{eqnarray}
{C_{\rm{3}}} &=&  - \mu \left[ {\frac{1}{{\sqrt {2\pi } }} + \frac{{{H_{\rm{E}}}}}{{{H_{\rm{B}}}}}\left( {\sqrt {\frac{{{H_{\rm{B}}}\varsigma _{\rm{B}}^2\sigma _{\rm{B}}^2}}{{2\pi M}}}  + \frac{{{H_{\rm{B}}}}}{2}\sqrt {\frac{{\xi P}}{M}} } \right)} \right] + \left( {\sqrt {\frac{{{H_{\rm{B}}}\varsigma _{\rm{B}}^2\sigma _{\rm{B}}^2}}{{2\pi M}}}  + \frac{{{H_{\rm{B}}}}}{2}\sqrt {\frac{{\xi P}}{M}} } \right) \nonumber\\
 &&\ge \sqrt {\frac{{{H_{\rm{B}}}\varsigma _{\rm{B}}^2\sigma _{\rm{B}}^2}}{{2\pi M}}}  + \frac{{{H_{\rm{B}}}}}{2}\sqrt {\frac{{\xi P}}{M}} .
\label{eq65}
\end{eqnarray}

\textbf{Case 2:} when $0 \le \mu  \le {H_{\rm{B}}}/{H_{\rm{E}}}$, ${C_{\rm{3}}}$ is given by
\begin{equation}
{C_{\rm{3}}} = \mu \left[ {\frac{1}{{\sqrt {2\pi } }} - \frac{{{H_{\rm{E}}}}}{{{H_{\rm{B}}}}}\left( {\sqrt {\frac{{{H_{\rm{B}}}\varsigma _{\rm{B}}^2\sigma _{\rm{B}}^2}}{{2\pi M}}}  + \frac{{{H_{\rm{B}}}}}{2}\sqrt {\frac{{\xi P}}{M}} } \right)} \right] + \left( {\sqrt {\frac{{{H_{\rm{B}}}\varsigma _{\rm{B}}^2\sigma _{\rm{B}}^2}}{{2\pi M}}}  + \frac{{{H_{\rm{B}}}}}{2}\sqrt {\frac{{\xi P}}{M}} } \right).
\label{eq66}
\end{equation}
If $\frac{1}{{\sqrt {2\pi } }} \ge \frac{{{H_{\rm{E}}}}}{{{H_{\rm{B}}}}}\left( {\sqrt {\frac{{{H_{\rm{B}}}\varsigma _{\rm{B}}^2\sigma _{\rm{B}}^2}}{{2\pi M}}}  + \frac{{{H_{\rm{B}}}}}{2}\sqrt {\frac{{\xi P}}{M}} } \right)$, we have
${C_3} \ge \sqrt {\frac{{{H_{\rm{B}}}\varsigma _{\rm{B}}^2\sigma _{\rm{B}}^2}}{{2\pi M}}}  + \frac{{{H_{\rm{B}}}}}{2}\sqrt {\frac{{\xi P}}{M}}$;
otherwise, we have
${C_3} \ge \frac{{{H_{\rm{B}}}}}{{{H_{\rm{E}}}\sqrt {2\pi } }}$.

\textbf{Case 3:} when $\mu  \ge {H_{\rm{B}}}/{H_{\rm{E}}}$, ${C_{\rm{3}}}$ is given by
\begin{eqnarray}
{C_3} &=& \mu \left[ {\frac{1}{{\sqrt {2\pi } }} + \frac{{{H_{\rm{E}}}}}{{{H_{\rm{B}}}}}\left( {\sqrt {\frac{{{H_{\rm{B}}}\varsigma _B^2\sigma _B^2}}{{2\pi M}}}  + \frac{{{H_{\rm{B}}}}}{2}\sqrt {\frac{{\xi P}}{M}} } \right)} \right] - \left( {\sqrt {\frac{{{H_{\rm{B}}}\varsigma _B^2\sigma _B^2}}{{2\pi M}}}  + \frac{{{H_{\rm{B}}}}}{2}\sqrt {\frac{{\xi P}}{M}} } \right) \nonumber\\
 &\ge& \frac{{{H_{\rm{B}}}}}{{{H_{\rm{E}}}\sqrt {2\pi } }}.
\label{eq67}
\end{eqnarray}

According to the above three cases, we have
\begin{eqnarray}
{C_3} \ge \left\{ \begin{array}{l}
\sqrt {\frac{{{H_{\rm{B}}}\varsigma _B^2\sigma _B^2}}{{2\pi M}}}  + \frac{{{H_{\rm{B}}}}}{2}\sqrt {\frac{{\xi P}}{M}} ,\;{\rm{if}}\;\frac{1}{{\sqrt {2\pi } }} \ge \frac{{{H_{\rm{E}}}}}{{{H_{\rm{B}}}}}\left( {\sqrt {\frac{{{H_{\rm{B}}}\varsigma _B^2\sigma _B^2}}{{2\pi M}}}  + \frac{{{H_{\rm{B}}}}}{2}\sqrt {\frac{{\xi P}}{M}} } \right)\\
\frac{{{H_{\rm{B}}}}}{{{H_{\rm{E}}}\sqrt {2\pi } }},\quad \quad\;\;\;\quad \quad \quad \;\;{\rm{otherwise}}
\end{array} \right..
\label{eq68}
\end{eqnarray}
Substituting (\ref{eq68}) into (\ref{eq64}), we obtain \emph{Theorem \ref{the2}}. \quad\quad\quad\quad\quad\quad\quad\quad\quad\quad\quad\quad\quad\quad\quad\;\;\;$\blacksquare$

\section{Proof of \emph{Corollary \ref{cor3}}}
\label{Appendix_C}
In \emph{Theorem \ref{the1}},
the term ${f_{{\rm{low}}}}\left( {{H_{\rm{B}}},\xi ,P} \right)$ is a monotonically decreasing positive function with respect to ${H_{\rm{B}}}$, $\xi $ and $P$ \cite{BIB34}.
When $P$ tends to infinity, we can obviously obtain
\begin{equation}
\mathop {\lim }\limits_{P \to \infty } {f_{{\rm{low}}}}\left( {{H_{\rm{B}}},\xi ,P} \right) = \ln \left( {{H_{\rm{B}}}} \right),
\label{eq69}
\end{equation}
and
\begin{eqnarray}
&&\mathop {\lim }\limits_{P \to \infty } {e^{\frac{1}{{{H_{\rm{B}}}\varsigma _{\rm{B}}^2\xi P}}}}Ei\left( { - \frac{1}{{{H_{\rm{B}}}\varsigma _{\rm{B}}^2\xi P}}} \right) - {e^{\frac{1}{{{H_{\rm{E}}}\varsigma _{\rm{E}}^2\xi P}}}}Ei\left( { - \frac{1}{{{H_{\rm{E}}}\varsigma _{\rm{E}}^2\xi P}}} \right) \nonumber\\
 &&= \mathop {\lim }\limits_{P \to \infty } {E_X}\left\{ {\ln \left( {\frac{{1 + {H_{\rm{E}}}\varsigma _{\rm{E}}^2X}}{{1 + {H_{\rm{B}}}\varsigma _{\rm{B}}^2X}}} \right)} \right\} \nonumber\\
 &&= \ln \left( {\frac{{{H_{\rm{E}}}\varsigma _{\rm{E}}^2}}{{{H_{\rm{B}}}\varsigma _{\rm{B}}^2}}} \right).
\label{eq70}
\end{eqnarray}
Substituting (\ref{eq69}) and (\ref{eq70}) into \emph{Theorem \ref{the1}}, we obtain the asymptotic lower bound on secrecy capacity as
\begin{eqnarray}
\mathop {\lim }\limits_{P \to \infty } {C_{{\rm{Low}}}}&{\rm{ = }}&\mathop {\lim }\limits_{P \to \infty } \left\{ {\frac{1}{2}\ln \left[ {\frac{{e{H_{\rm{B}}}{H_{\rm{E}}}\varsigma _{\rm{E}}^2\sigma _{\rm{E}}^2{\xi ^2}{P^2}}}{{2\pi \varsigma _{\rm{B}}^2\sigma _{\rm{B}}^2\left( {H_{\rm{E}}^2{\xi ^2}{P^2} + {H_{\rm{E}}}\xi P\varsigma _{\rm{E}}^2\sigma _{\rm{E}}^2 + \sigma _{\rm{E}}^2} \right)}}} \right]} \right\} \nonumber\\
&{\rm{ = }}&\frac{1}{2}\ln \left( {\frac{{e{H_{\rm{B}}}\varsigma _{\rm{E}}^2\sigma _{\rm{E}}^2}}{{2\pi {H_{\rm{E}}}\varsigma _{\rm{B}}^2\sigma _{\rm{B}}^2}}} \right).
\label{eq71}
\end{eqnarray}

For the asymptotic upper bound, the case of $\frac{1}{{\sqrt {2\pi } }} \ge \frac{{{H_{\rm{E}}}}}{{{H_{\rm{B}}}}}\left( {\sqrt {\frac{{{H_{\rm{B}}}\varsigma _B^2\sigma _B^2}}{{2\pi M}}}  + \frac{{{H_{\rm{B}}}}}{2}\sqrt {\frac{{\xi P}}{M}} } \right)$ does not occur when $P$ tends to infinity.
Therefore, we can obtain the asymptotic upper bound as
\begin{equation}
\mathop {\lim }\limits_{P \to \infty } {C_{{\rm{Upp}}}}{\rm{ = }}\frac{1}{2}\ln \left( {\frac{{4e{H_{\rm{B}}}\varsigma _{\rm{E}}^2\sigma _{\rm{E}}^2}}{{{\pi ^2}{H_{\rm{E}}}\varsigma _{\rm{B}}^2\sigma _{\rm{B}}^2}}} \right).
\label{eq72}
\end{equation}
Combining (\ref{eq71}) with (\ref{eq72}), we prove \emph{Corollary \ref{cor3}}. \quad\quad\quad\quad\quad\quad\quad\quad\quad\quad\quad\quad\quad\quad\quad\quad\;\;\;$\blacksquare$

\section{Proof of \emph{Theorem \ref{the3}}}
\label{Appendix_D}
Similar to (\ref{eq35}), the functional optimization problem in this case can be expressed as
\begin{eqnarray}
&& \mathop {\max }\limits_{{f_X}\left( x \right)} {\cal H}( X ) =  - \int_0^A {{f_X}( x ){\rm{ln}}\left[ {{f_X}( x )} \right]{\rm{d}}x} \nonumber\\
{\rm{s}}{\rm{.t}}{\rm{.}}\;&&\int_0^A {{f_X}( x ){\rm{d}}x}  = 1 \label{eq73}\\
&&E\left( X \right) = \int_0^A {x{f_X}( x ){\rm{d}}x}  = \xi P. \nonumber
\end{eqnarray}
By employing the variational method, we can derive the input PDF as \cite{BIB28}
\begin{equation}
{f_X}(x) = {e^{cx + b - 1}},
\label{eq74}
\end{equation}
where $b$ can $c$ are two free parameters.

When $c = 0$, substitute (\ref{eq74}) into the constraints in (\ref{eq73}), we can obtain the input PDF as \cite{BIB28}
\begin{eqnarray}
{f_X}(x) = \left\{ \begin{array}{l}
\frac{1}{A},\;x \in [0,A]\\
0,\;\;{\rm{otherwise}}
\end{array} \right..
\label{eq75}
\end{eqnarray}
In this case, ${E_X}\left( X \right) = A/2 = \xi P$ , we have $\alpha  = 0.5$.
Then, we can obtain ${\cal H}\left( X \right)$, ${\mathop{\rm var}} \left( {{Y_{\rm{E}}}} \right)$ and ${E_X}\left[ {\ln \left( {\frac{{1 + {H_{\rm{E}}}\varsigma _{\rm{E}}^2X}}{{1 + {H_{\rm{B}}}\varsigma _{\rm{B}}^2X}}} \right)} \right]$ in (\ref{eq13}) as
\begin{equation}
{\cal H}(X) = \ln (A),
\label{eq76}
\end{equation}
\begin{eqnarray}
{\mathop{\rm var}} \left( {{Y_{\rm{E}}}} \right) 
 &=& H_{\rm{E}}^2{\rm{var}}\left( X \right) + {H_{\rm{E}}}{\mathop{\rm var}} \left( {\sqrt X {Z_{{\rm{E}},1}}} \right){\rm{ + }}{\mathop{\rm var}} \left( {{Z_{{\rm{E,0}}}}} \right) \nonumber\\
 &=& \frac{{H_{\rm{E}}^2{A^2}}}{{12}} + \frac{A}{2}{H_{\rm{E}}}\varsigma _{\rm{E}}^2\sigma _{\rm{E}}^2 + \sigma _{\rm{E}}^2,
\label{eq77}
\end{eqnarray}
and
\begin{eqnarray}
{E_X}\left[ {\ln \left( {\frac{{1 + {H_{\rm{E}}}\varsigma _{\rm{E}}^2X}}{{1 + {H_{\rm{B}}}\varsigma _{\rm{B}}^2X}}} \right)} \right] \!\!\!\!\!&=&\!\!\!\!\! \int_0^A {\frac{1}{A} \cdot \ln \left( {\frac{{1 + {H_{\rm{E}}}\varsigma _{\rm{E}}^2x}}{{1 + {H_{\rm{B}}}\varsigma _{\rm{B}}^2x}}} \right){\rm{d}}x} \nonumber \\
 &=&\!\!\!\!\! \frac{1}{A}\!\left[\! {A\ln \!\left(\! {\frac{{1 + {H_{\rm{E}}}\varsigma _{\rm{E}}^2A}}{{1 + {H_{\rm{B}}}\varsigma _{\rm{B}}^2A}}} \!\right) \!-\! \int_0^A {\!\left( {\frac{1}{{1 + {H_{\rm{B}}}\varsigma _{\rm{B}}^2x}} - \frac{1}{{1 + {H_{\rm{E}}}\varsigma _{\rm{E}}^2x}}} \!\right)\!{\rm{d}}x} } \!\right] \nonumber \\
 &=&\!\!\!\!\! \ln \left( {\frac{{1 + {H_{\rm{E}}}\varsigma _{\rm{E}}^2A}}{{1 + {H_{\rm{B}}}\varsigma _{\rm{B}}^2A}}} \right) - \frac{{\ln \left( {1 + {H_{\rm{B}}}\varsigma _{\rm{B}}^2A} \right)}}{{A{H_{\rm{B}}}\varsigma _{\rm{B}}^2}} + \frac{{\ln \left( {1 + {H_{\rm{E}}}\varsigma _{\rm{E}}^2A} \right)}}{{A{H_{\rm{E}}}\varsigma _{\rm{E}}^2}}.
\label{eq78}
\end{eqnarray}
Substituting (\ref{eq76}), (\ref{eq77}) and (\ref{eq78}) into (\ref{eq13}), we obtain the lower bound on secrecy capacity for $\alpha  = 0.5$.

When $c \ne 0$, we have $\alpha  \ne 0.5$.
Substituting (\ref{eq74}) into the constraints in (\ref{eq73}), we can derive the input PDF as \cite{BIB28}
\begin{eqnarray}
{f_X}(x) = \left\{ \begin{array}{l}
\frac{{c{e^{cx}}}}{{{e^{cA}} - 1}},x \in [0,A]\\
0,\;\;\;\,\quad {\rm{otherwise}}
\end{array} \right.,
\label{eq79}
\end{eqnarray}
where $c$ is the solution to (\ref{eq29}). Then, we can obtain
\begin{eqnarray}
H\left( X \right) &=&  - \int_0^A {{f_X}\left( x \right)\ln \left( {\frac{{c{e^{cx}}}}{{{e^{cA}} - 1}}} \right){\rm{d}}x} \nonumber \\
 &=& \ln \left( {\frac{{{e^{cA}} - 1}}{c}} \right) - c\xi P,
\label{eq80}
\end{eqnarray}
\begin{eqnarray}
{\mathop{\rm var}} \left( {{Y_{\rm{E}}}} \right) &=& H_{\rm{E}}^2{\rm{var}}\left( X \right) + {H_{\rm{E}}}{\mathop{\rm var}} \left( {\sqrt X {Z_{{\rm{E}},1}}} \right){\rm{ + }}{\mathop{\rm var}} \left( {{Z_{{\rm{E,0}}}}} \right) \nonumber\\
 &=& H_{\rm{E}}^2\left[ {\frac{{A\left( {cA - 2} \right)}}{{c\left( {1 - {e^{ - cA}}} \right)}} + \frac{2}{{{c^2}}} - {\xi ^2}{P^2}} \right] + {H_{\rm{E}}}\xi P\varsigma _{\rm{E}}^2\sigma _{\rm{E}}^2 + \sigma _{\rm{E}}^2,
\label{eq81}
\end{eqnarray}
and
\begin{eqnarray}
{E_X}\left[ {\ln \left( {\frac{{1 + {H_{\rm{E}}}\varsigma _{\rm{E}}^2X}}{{1 + {H_{\rm{B}}}\varsigma _{\rm{B}}^2X}}} \right)} \right] \!\!\!\!\!&=&\!\!\!\!\! \int_0^A {\ln \left( {\frac{{1 + {H_{\rm{E}}}\varsigma _{\rm{E}}^2x}}{{1 + {H_{\rm{B}}}\varsigma _{\rm{B}}^2x}}} \right) \cdot \frac{{c{e^{cx}}}}{{{e^{cA}} - 1}}{\rm{d}}x} \nonumber \\
 &=&\!\!\!\!\! \frac{1}{{{e^{cA}} - 1}}\left\{ {\ln \left( {\frac{{1 + {H_{\rm{E}}}\varsigma _{\rm{E}}^2A}}{{1 + {H_{\rm{B}}}\varsigma _{\rm{B}}^2A}}} \right){e^{cA}}} \right. \nonumber \\
 &&\!\!\!\!\!- {e^{ - \frac{c}{{{H_{\rm{E}}}\varsigma _{\rm{E}}^2}}}}\left[ {Ei\left( {\frac{{c\left( {1 + {H_{\rm{E}}}\varsigma _{\rm{E}}^2A} \right)}}{{{H_{\rm{E}}}\varsigma _{\rm{E}}^2}}} \right) - Ei\left( {\frac{c}{{{H_{\rm{E}}}\varsigma _{\rm{E}}^2}}} \right)} \right] \nonumber \\
&&\!\!\!\!\!\left. { + {e^{ - \frac{c}{{{H_{\rm{B}}}\varsigma _{\rm{B}}^2}}}}\left[ {Ei\left( {\frac{{c\left( {1 + {H_{\rm{B}}}\varsigma _{\rm{B}}^2A} \right)}}{{{H_{\rm{B}}}\varsigma _{\rm{B}}^2}}} \right) - Ei\left( {\frac{c}{{{H_{\rm{B}}}\varsigma _{\rm{B}}^2}}} \right)} \right]} \right\}.
\label{eq82}
\end{eqnarray}
Substituting (\ref{eq80}), (\ref{eq81}) and (\ref{eq82}) into (\ref{eq13}), we obtain the lower bound on secrecy capacity for $\alpha  \ne 0.5$. \quad\quad\quad\quad\quad\quad\quad\quad\quad\quad\quad\quad\quad\quad\quad\quad\quad\quad\quad\quad\quad\quad\quad\quad\quad\quad\quad\quad\quad\quad\quad\quad\quad\quad\;\;$\blacksquare$

\section{Proof of \emph{Theorem \ref{the4}}}
\label{Appendix_E}
\emph{1) Proof of monotonicity:} When $\alpha  = 0.5$, the conclusion is obvious.
In the following, we will prove the conclusion when $\alpha  \ne 0.5$.
When $\alpha  \ne 0.5$, we let $F\left( c \right) = \frac{1}{{1 - {e^{ - cA}}}} - \frac{1}{{cA}},\,c \ne 0$ in (\ref{eq29}).
Taking the first-order derivative of $F\left( c \right)$, we can obtain
\begin{eqnarray}
\frac{{{\rm{d}}F\left( c \right)}}{{{\rm{d}}c}} 
 &=& \frac{{\left( {1 - {e^{ - cA}} + cA{e^{ - \frac{{cA}}{2}}}} \right)\left( {1 - {e^{ - cA}} - cA{e^{ - \frac{{cA}}{2}}}} \right)}}{{{c^2}A{{\left( {1 - {e^{ - cA}}} \right)}^2}}}.
\label{eq83}
\end{eqnarray}
When $c > 0$, we have $1 - {e^{ - cA}} + cA{e^{ - \frac{{cA}}{2}}} > 0$ and ${c^2}A{\left( {1 - {e^{ - cA}}} \right)^2} > 0$ .
Then, we let $G(c) = 1 - {e^{ - cA}} - cA{e^{ - \frac{{cA}}{2}}},\,c > 0$. Whether the value of (\ref{eq83}) is positive or negative depends on $G(c)$.
Taking the first-order derivative of $G(c)$, we have
\begin{eqnarray}
\frac{{{\rm{d}}G\left( c \right)}}{{{\rm{d}}c}} 
 &=& A{e^{ - \frac{{cA}}{2}}}\left( {{e^{ - \frac{{cA}}{2}}} - 1 + \frac{{cA}}{2}} \right).
\label{eq84}
\end{eqnarray}
Let $K(c) = {e^{ - \frac{{cA}}{2}}} - 1 + \frac{{cA}}{2},\,c > 0$, we have
\begin{eqnarray}
\frac{{{\rm{d}}K\left( c \right)}}{{{\rm{d}}c}} 
 &=&  - \frac{A}{2}{e^{ - \frac{{cA}}{2}}} + \frac{A}{2} > 0.
\label{eq85}
\end{eqnarray}
This indicates that $K(c)$ is an increasing function of $c$.
Moreover, we have $K(0) = 0$. As a result, $K(c) > 0,\,\forall c > 0$.
Then, in (\ref{eq84}), we have ${\rm{d}}G\left( c \right)/{\rm{d}}c > 0$ and $G(0) = 0$, and thus $G(c) > 0,\,\forall c > 0$.
After that, in (\ref{eq83}), we have ${\rm{d}}F\left( c \right)/{\rm{d}}c > 0$, which indicates that $F\left( c \right)$ is an increasing function of $c$ when $c > 0$.
Moreover, by using L'Hospita's rule, we have
\begin{eqnarray}
\mathop {\lim }\limits_{c \to 0} F\left( c \right) &=& \mathop {\lim }\limits_{c \to 0} \left( {\frac{{cA - 1 + {e^{ - cA}}}}{{cA\left( {1 - {e^{ - cA}}} \right)}}} \right) \nonumber \\
 &=& 0.5.
\label{eq86}
\end{eqnarray}

Therefore, we have
\begin{equation}
F\left( c \right) = \alpha  > 0.5,\,{\rm{if}}\,c > 0.
\label{eq87}
\end{equation}
According to (\ref{eq87}), we know that when $\alpha  \in (0.5,1]$, $c > 0$ and thus the PDF (\ref{eq79}) is an increasing function of $x$ in $[0,A]$.
By using the similar analysis approach, we can also prove that when $\alpha  \in (0,0.5)$, $c < 0$ and thus the PDF (\ref{eq79}) is a decreasing function of $x$ in $[0,A]$.

\emph{2) Proof of symmetry:} For a fixed $\alpha $ value, the input PDF in (\ref{eq79}) is re-denoted by ${f_X}(x,\alpha )$ to facilitate the notation.
When $x \in [0,A]$, we have
\begin{eqnarray}
{f_X}\left( {A - x,1 - \alpha } \right) &=& \frac{{{c_2}{e^{{c_2}\left( {A - x} \right)}}}}{{{e^{{c_2}A}} - 1}} \nonumber\\
 &=& \frac{{\left( { - {c_2}} \right){e^{\left( { - {c_2}} \right)x}}}}{{{e^{\left( { - {c_2}} \right)A}} - 1}},
\label{eq88}
\end{eqnarray}
where ${c_2}$ is the solution of the following equation
\begin{equation}
1 - \alpha  = \frac{1}{{1 - {e^{\left( { - {c_2}} \right)A}}}} + \frac{1}{{\left( { - {c_2}} \right)A}}.
\label{eq89}
\end{equation}
Then, eq. (\ref{eq89}) can be further written as
\begin{eqnarray}
\alpha  &=& 1 - \frac{1}{{1 - {e^{\left( { - {c_2}} \right)A}}}} - \frac{1}{{\left( { - {c_2}} \right)A}} \nonumber \\
 &=& \frac{1}{{1 - {e^{{c_2}A}}}} - \frac{1}{{\left( { - {c_2}} \right)A}}.
\label{eq90}
\end{eqnarray}
According to (\ref{eq29}) and (\ref{eq90}), we have ${c_2} =  - c$.
Then, eq. (\ref{eq88}) can be further expressed as
\begin{equation}
{f_X}\left( {A - x,1 - \alpha } \right) = \frac{{c{e^{cx}}}}{{{e^{cA}} - 1}} = {f_X}\left( {x,\alpha } \right).
\label{eq91}
\end{equation}
This indicates that the curves of maxentropic input PDF (\ref{eq79}) with $\alpha $ and $1 - \alpha $ are symmetric with respect to $X = A/2$. \quad\quad\quad\quad\quad\quad\quad\quad\quad\quad\quad\quad\quad\quad\quad\quad\quad\quad\quad\quad\quad\quad\quad\quad\quad\quad\quad\quad\quad\;\;$\blacksquare$

\section{Proof of \emph{Theorem \ref{the5}}}
\label{Appendix_F}
In this scenario, eq. (\ref{eq44}) also holds, where ${I_1}$ can also be expressed as (\ref{eq51}).
To obtain ${I_2}$ in (\ref{eq44}), we select ${g_{{Y_{\rm{B}}}|{Y_{\rm E}}}}({y_{\rm{B}}}|{y_{\rm{E}}})$ as
\begin{equation}
{g_{{Y_{\rm{B}}}|{Y_{\rm E}}}}({y_{\rm{B}}}|{y_{\rm{E}}}) = \frac{1}{{\sqrt {2\pi {s^2}} }}{e^{ - \frac{{{{\left( {{y_{\rm{B}}} - \mu {y_{\rm{E}}}} \right)}^2}}}{{2{s^2}}}}},
\label{eq92}
\end{equation}
where $\mu $ and $s$ are free parameters to be determined.
Then, ${I_2}$ is given by
\begin{eqnarray}
{I_2} &=&  - {E_{{X^*}{Y_{\rm{B}}}}}\left\{ {\int_{ - \infty }^\infty  {{f_{{Y_{\rm{E}}}|X{Y_{\rm{B}}}}}({y_{\rm{E}}}|X,{Y_{\rm{B}}})\ln \left[ {\frac{{{e^{ - \frac{{{{\left( {{Y_{\rm{B}}} - \mu {y_{\rm{E}}}} \right)}^2}}}{{2{s^2}}}}}}}{{\sqrt {2\pi {s^2}} }}} \right]{\rm{d}}{y_{\rm{E}}}} } \right\} \nonumber\\
 &=& {E_{{X^*}{Y_{\rm{B}}}}}\left[ {\frac{1}{2}\ln \left( {2\pi {s^2}} \right) + \frac{1}{{2{s^2}}}\underbrace {\int_{ - \infty }^\infty  {\frac{{{e^{ - \frac{{{{\left( {{y_{\rm{E}}} - \frac{{{H_{\rm{E}}}}}{{{H_{\rm{B}}}}}{Y_{\rm{B}}}} \right)}^2}}}{{2\left( {MX + N} \right)}}}}}}{{\sqrt {2\pi \left( {MX + N} \right)} }}{{\left( {{Y_{\rm{B}}} - \mu {y_{\rm{E}}}} \right)}^2}{\rm{d}}{y_{\rm{E}}}} }_{{C_4}}} \right],
\label{eq93}
\end{eqnarray}
where ${C_4}$ can be expressed as
\begin{eqnarray}
{C_4} &=& \int_{ - \infty }^\infty  {\frac{1}{{\sqrt {2\pi } }}{e^{ - \frac{{{t^2}}}{2}}}{{\left[ {\left( {1 - \mu \frac{{{H_{\rm{E}}}}}{{{H_{\rm{B}}}}}} \right){Y_{\rm{B}}} - \mu \sqrt {MX + N} t} \right]}^2}{\rm{d}}t} \nonumber \\
 &=& {\left( {1 - \mu \frac{{{H_{\rm{E}}}}}{{{H_{\rm{B}}}}}} \right)^2}Y_{\rm{B}}^2 + {\mu ^2}\left( {MX + N} \right).
\label{eq94}
\end{eqnarray}

Substituting (\ref{eq94}) into (\ref{eq93}), we can obtain ${I_2}$ as
\begin{eqnarray}
{I_2} \!\!\!\!\!
 &=&\!\!\!\!\! {E_{{X^*}}}\!\left\{\! {\frac{1}{2}\ln \!\left( {2\pi {s^2}} \right) \!+\! \frac{1}{{2{s^2}}}\!\!\left[\! {{{\left(\! {1 \!-\! \mu \frac{{{H_{\rm{E}}}}}{{{H_{\rm{B}}}}}} \!\right)}^2}\!\int_{ - \infty }^\infty  {\frac{{y_{\rm{B}}^2{e^{ - \frac{{{{\left( {{y_{\rm{B}}} - {H_{\rm{B}}}X} \right)}^2}}}{{2\left( {1 \!+\! {H_{\rm{B}}}X\varsigma _{\rm{B}}^2} \right)\sigma _{\rm{B}}^2}}}}}}{{\sqrt {2\pi \left( {1 \!+\! {H_{\rm{B}}}X\varsigma _{\rm{B}}^2} \right)\sigma _{\rm{B}}^2} }}{\rm{d}}{y_{\rm{B}}}}  \!+\! {\mu ^2}\!\left(\! {MX \!+\! N} \!\right)} \!\right]} \!\right\} \nonumber \\
 &=&\!\!\!\!\! {E_{{X^*}}}\!\left\{\! {\frac{1}{2}\ln \!\left( {2\pi {s^2}} \right) \!+\! \frac{1}{{2{s^2}}}\!\left[\! {{{\left(\! {1 \!-\! \mu \frac{{{H_{\rm{E}}}}}{{{H_{\rm{B}}}}}} \!\right)}^2}\!\left[ {\left( {1 \!+\! {H_{\rm{B}}}X\varsigma _{\rm{B}}^2} \right)\sigma _{\rm{B}}^2 \!+\! H_{\rm{B}}^2{X^2}} \right] \!+\! {\mu ^2}\!\left(\! {MX \!+\! N} \!\right)} \!\right]} \!\right\} \nonumber \\
 &\le&\!\!\!\!\! {E_{{X^*}}}\!\left\{ {\underbrace {\frac{1}{2}\ln \!\left( {2\pi {s^2}} \right) \!+\! \frac{1}{{2{s^2}}}\!\left[\! {{{\left(\! {1 \!-\! \mu \frac{{{H_{\rm{E}}}}}{{{H_{\rm{B}}}}}} \right)}^2}\!\!\left( {{H_{\rm{B}}}^2AX \!+\! {H_{\rm{B}}}X\varsigma _{\rm{B}}^2\sigma _{\rm{B}}^2 \!+\! \sigma _{\rm{B}}^2} \!\right) \!+\! {\mu ^2}\!\left( {MX \!+\! N} \!\right)} \!\right]}_{{C_5}}} \!\right\}\!,
\label{eq95}
\end{eqnarray}
where the last inequality holds because $X \le A$.

Taking the first partial derivative of ${C_5}$ with respect to ${s^2}$, and let it be zero, we can obtain the minimum point as
\begin{equation}
{s^2} = {\left( {1 - \mu \frac{{{H_{\rm{E}}}}}{{{H_{\rm{B}}}}}} \right)^2}\left( {{H_{\rm{B}}}^2AX + {H_{\rm{B}}}X\varsigma _{\rm{B}}^2\sigma _{\rm{B}}^2 + \sigma _{\rm{B}}^2} \right) + {\mu ^2}\left( {MX + N} \right).
\label{eq96}
\end{equation}
Then, substituting (\ref{eq96}) into (\ref{eq95}), we can rewrite ${C_5}$ as
\begin{equation}
{C_5} = \frac{1}{2}\ln \left\{ {2\pi e\left[ {{{\left( {1 - \mu \frac{{{H_{\rm{E}}}}}{{{H_{\rm{B}}}}}} \right)}^2}\left( {{H_{\rm{B}}}^2AX + {H_{\rm{B}}}X\varsigma _{\rm{B}}^2\sigma _{\rm{B}}^2 + \sigma _{\rm{B}}^2} \right) + {\mu ^2}\left( {MX + N} \right)} \right]} \right\}.
\label{eq97}
\end{equation}

Substituting (\ref{eq51}), (\ref{eq95}) and (\ref{eq97}) into (\ref{eq44}), we can obtain the upper bound on secrecy capacity as
\begin{eqnarray}
{C_s} \!\!\!&\le&\!\!\!  - \frac{1}{2}\ln \left( {2\pi e\frac{{\sigma _{\rm{B}}^2}}{{\sigma _{\rm{E}}^2}}} \right) - \frac{1}{2}{E_{{X^*}}}\left\{ {\ln \left( {MX + N} \right) + \ln \left( {\frac{{1 + {H_{\rm{B}}}X\varsigma _{\rm{B}}^2}}{{1 + {H_{\rm{E}}}X\varsigma _{\rm{E}}^2}}} \right)} \right\} \nonumber \\
 &&\!\!\!+ \frac{1}{2}{E_{{X^*}}}\left\{ {\ln \left[ {2\pi e\left( {{{\left( {1 - \mu \frac{{{H_{\rm{E}}}}}{{{H_{\rm{B}}}}}} \right)}^2}\left( {{H_{\rm{B}}}^2AX + {H_{\rm{B}}}X\varsigma _{\rm{B}}^2\sigma _{\rm{B}}^2 + \sigma _{\rm{B}}^2} \right) + {\mu ^2}\left( {MX + N} \right)} \!\right)} \!\right]} \!\right\} \nonumber \\
 &=&\!\!\!  - \frac{1}{2}\ln \left( {2\pi e\frac{{\sigma _{\rm{B}}^2}}{{\sigma _{\rm{E}}^2}}} \right) - \frac{1}{2}{E_{{X^*}}}\left\{ {\ln \left( {\frac{{1 + {H_{\rm{B}}}X\varsigma _{\rm{B}}^2}}{{1 + {H_{\rm{E}}}X\varsigma _{\rm{E}}^2}}} \right)} \right\} \nonumber \\
 &&\!\!\!+ \frac{1}{2}{E_{{X^*}}}\left\{ {\ln \left[ {2\pi e\left( {{{\left( {1 - \mu \frac{{{H_{\rm{E}}}}}{{{H_{\rm{B}}}}}} \right)}^2}\frac{{{H_{\rm{B}}}^2AX + {H_{\rm{B}}}X\varsigma _{\rm{B}}^2\sigma _{\rm{B}}^2 + \sigma _{\rm{B}}^2}}{{MX + N}} + {\mu ^2}} \right)} \right]} \right\}.
\label{eq98}
\end{eqnarray}
By Jenson's inequality for the convex function $\ln ( \cdot )$, we further have
\begin{eqnarray}
{C_s} &\le&  - \frac{1}{2}\ln \left( {2\pi e\frac{{\sigma _{\rm{B}}^2}}{{\sigma _{\rm{E}}^2}}} \right) - \frac{1}{2}\ln \left[ {{E_{{X^*}}}\left( {\frac{{1 + {H_{\rm{B}}}X\varsigma _{\rm{B}}^2}}{{1 + {H_{\rm{E}}}X\varsigma _{\rm{E}}^2}}} \right)} \right] \nonumber \\
 &&+ \frac{1}{2}\ln \left[ {2\pi e\left( {{{\left( {1 - \mu \frac{{{H_{\rm{E}}}}}{{{H_{\rm{B}}}}}} \right)}^2}{E_{{X^*}}}\left( {\frac{{{H_{\rm{B}}}^2AX + {H_{\rm{B}}}X\varsigma _{\rm{B}}^2\sigma _{\rm{B}}^2 + \sigma _{\rm{B}}^2}}{{MX + N}}} \right) + {\mu ^2}} \right)} \right].
\label{eq99}
\end{eqnarray}
Because we are more interested in the performance at large optical intensity.
As the peak optical intensity $A$ is loosened to infinity, the secrecy-capacity-achieving input PDF will escape to infinity \cite{BIB16}, \cite{BIB34}.
This indicates that when $A$ approaches infinity, the probability for any set of finite-intensity input symbols will tend to zero.
Therefore, when $A$ tends to infinity, we have
\begin{eqnarray}
\left\{ \begin{array}{l}
{E_{{X^*}}}\left( {\frac{{1 + {H_{\rm{B}}}X\varsigma _{\rm{B}}^2}}{{1 + {H_{\rm{E}}}X\varsigma _{\rm{E}}^2}}} \right) =  \frac{{1 + {H_{\rm{B}}}\varsigma _{\rm{B}}^2A}}{{1 + {H_{\rm{E}}}\varsigma _{\rm{E}}^2A}}\\
{E_{{X^*}}}\left( {\frac{{{H_{\rm{B}}}^2AX + {H_{\rm{B}}}X\varsigma _{\rm{B}}^2\sigma _{\rm{B}}^2 + \sigma _{\rm{B}}^2}}{{MX + N}}} \right) = \frac{{H_{\rm{B}}^2 A^2 + {H_{\rm{B}}}A\varsigma _{\rm{B}}^2\sigma _{\rm{B}}^2 + \sigma _{\rm{B}}^2}}{{MA + N}}
\end{array} \right..
\label{eq100}
\end{eqnarray}
By using L'Hospital's rule, we have
\begin{eqnarray}
\left\{ \begin{array}{l}
\mathop {\lim }\limits_{A \to \infty } \frac{{1 + {H_{\rm{B}}}\varsigma _{\rm{B}}^2A}}{{1 + {H_{\rm{E}}}\varsigma _{\rm{E}}^2A}} = \frac{{{H_{\rm{B}}}\varsigma _{\rm{B}}^2}}{{{H_{\rm{E}}}\varsigma _{\rm{E}}^2}}\\
\mathop {\lim }\limits_{A \to \infty } \frac{{H_{\rm{B}}^2A^2 + {H_{\rm{B}}}A\varsigma _{\rm{B}}^2\sigma _{\rm{B}}^2 + \sigma _{\rm{B}}^2}}{{MA + N}} = \frac{{H_{\rm{B}}^2A + {H_{\rm{B}}}\varsigma _{\rm{B}}^2\sigma _{\rm{B}}^2}}{M}
\end{array} \right..
\label{eq101}
\end{eqnarray}
According to (\ref{eq99}), (\ref{eq100}) and (\ref{eq101}), the upper bound of secrecy capacity can be asymptotically expressed as
\begin{equation}
{C_s} \le  - \frac{1}{2}\ln \left( {2\pi e\frac{{{H_{\rm{B}}}\varsigma _{\rm{B}}^2\sigma _{\rm{B}}^2}}{{{H_{\rm{E}}}\varsigma _{\rm{E}}^2\sigma _{\rm{E}}^2}}} \right) + \frac{1}{2}\ln \left[ {2\pi e\underbrace {\left( {{{\left( {1 - \mu \frac{{{H_{\rm{E}}}}}{{{H_{\rm{B}}}}}} \right)}^2}\frac{{H_{\rm{B}}^2A + {H_{\rm{B}}}\varsigma _{\rm{B}}^2\sigma _{\rm{B}}^2}}{M} + {\mu ^2}} \right)}_{{C_6}}} \right].
\label{eq102}
\end{equation}

Taking the first partial derivative of ${C_6}$ with respect to $\mu $, and let it to be zero, we have
\begin{equation}
\mu  = \frac{{\frac{{{H_{\rm{E}}}}}{{{H_{\rm{B}}}}}\frac{{H_{\rm{B}}^2A + {H_{\rm{B}}}\varsigma _{\rm{B}}^2\sigma _{\rm{B}}^2}}{M}}}{{1 + {{\left( {\frac{{{H_{\rm{E}}}}}{{{H_{\rm{B}}}}}} \right)}^2}\frac{{H_{\rm{B}}^2A + {H_{\rm{B}}}\varsigma _{\rm{B}}^2\sigma _{\rm{B}}^2}}{M}}}.
\label{eq103}
\end{equation}
Substituting (\ref{eq103}) into (\ref{eq102}), we can derive (\ref{eq31}). \quad\quad\quad\quad\quad\quad\quad\quad\quad\quad\quad\quad\quad\quad\quad\quad\quad\;\;$\blacksquare$

\section{Proof of \emph{Corollary \ref{cor6}}}
\label{Appendix_G}
The average to peak optical intensity ratio $\alpha  \in \left( {0,1} \right]$ is a positive constant, and $\alpha A = \xi P$.
When $A$ tends to infinity, $P$ also tends to infinity.
When $\alpha  = 0.5$, we have
\begin{equation}
\mathop {\lim }\limits_{A \to \infty } {f_{{\rm{low}}}}\left( {{H_{\rm{B}}},\xi ,P} \right) = \mathop {\lim }\limits_{P \to \infty } {f_{{\rm{low}}}}\left( {{H_{\rm{B}}},\xi ,P} \right) = \ln \left( {{H_{\rm{B}}}} \right).
\label{eq104}
\end{equation}
By using L'Hospital's rule, we have
\begin{eqnarray}
\left\{ \begin{array}{l}
\mathop {\lim }\limits_{A \to \infty } \frac{1}{2}\ln \left( {\frac{{1 + {H_{\rm{E}}}\varsigma _{\rm{E}}^2A}}{{1 + {H_{\rm{B}}}\varsigma _{\rm{B}}^2A}}} \right) = \frac{1}{2}\ln \left( {\frac{{{H_{\rm{E}}}\varsigma _{\rm{E}}^2}}{{{H_{\rm{B}}}\varsigma _{\rm{B}}^2}}} \right)\\
\mathop {\lim }\limits_{A \to \infty } \frac{{\ln \left( {1 + {H_{\rm{B}}}\varsigma _{\rm{B}}^2A} \right)}}{{2A{H_{\rm{B}}}\varsigma _{\rm{B}}^2}} = 0\\
\mathop {\lim }\limits_{A \to \infty } \frac{{\ln \left( {1 + {H_{\rm{E}}}\varsigma _{\rm{E}}^2A} \right)}}{{2A{H_{\rm{E}}}\varsigma _{\rm{E}}^2}} = 0\\
\mathop {\lim }\limits_{A \to \infty } \frac{1}{2}\ln \left( {\frac{{6{A^2}\sigma _{\rm{E}}^2}}{{\pi e\sigma _{\rm{B}}^2\left( {H_{\rm{E}}^2{A^2} + 6A{H_{\rm{E}}}\varsigma _{\rm{E}}^2\sigma _{\rm{E}}^2 + 12\sigma _{\rm{E}}^2} \right)}}} \right) = \frac{1}{2}\ln \left( {\frac{{6\sigma _{\rm{E}}^2}}{{\pi e\sigma _{\rm{B}}^2H_{\rm{E}}^2}}} \right).
\end{array} \right.
\label{eq105}
\end{eqnarray}
According (\ref{eq104}) and (\ref{eq105}), we can obtain the asymptotic lower bound when $\alpha  = 0.5$ as
\begin{equation}
\mathop {\lim }\limits_{A \to \infty } {C'_{{\rm{Low}}}} = \frac{1}{2}\ln \left( {\frac{{6{H_{\rm{B}}}\varsigma _{\rm{E}}^2\sigma _{\rm{E}}^2}}{{\pi e{H_{\rm{E}}}\varsigma _{\rm{B}}^2\sigma _{\rm{B}}^2}}} \right).
\label{eq106}
\end{equation}

When $\alpha  \ne 0.5$, we can also obtain (\ref{eq104}).
Moreover, according to (\ref{eq82}), (\ref{eq100}) and (\ref{eq101}), we have
\begin{eqnarray}
&&\mathop {\lim }\limits_{A \to \infty } \frac{1}{{2\left( {{e^{cA}} - 1} \right)}}\!\left\{ \!{\ln \!\left( {\frac{{1 + {H_{\rm{E}}}\varsigma _{\rm{E}}^2A}}{{1 + {H_{\rm{B}}}\varsigma _{\rm{B}}^2A}}} \!\right)\!{e^{cA}}} \right. \!-\! {e^{ - \frac{c}{{{H_{\rm{E}}}\varsigma _{\rm{E}}^2}}}}\!\left[ \!{Ei\!\left(\! {\frac{c}{{{H_{\rm{E}}}\varsigma _{\rm{E}}^2}}\left( {1 \!+\! {H_{\rm{E}}}\varsigma _{\rm{E}}^2A} \right)} \!\right) \!-\! Ei\!\left(\! {\frac{c}{{{H_{\rm{E}}}\varsigma _{\rm{E}}^2}}} \!\right)} \!\right] \nonumber \\
&&\quad\left. { + {e^{ - \frac{c}{{{H_{\rm{B}}}\varsigma _{\rm{B}}^2}}}}\left[ {Ei\left( {\frac{c}{{{H_{\rm{B}}}\varsigma _{\rm{B}}^2}}\left( {1 + {H_{\rm{B}}}\varsigma _{\rm{B}}^2A} \right)} \right) - Ei\left( {\frac{c}{{{H_{\rm{B}}}\varsigma _{\rm{B}}^2}}} \right)} \right]} \right\} \nonumber \\
 &&= \mathop {\lim }\limits_{A\to \infty } \frac{1}{2}{E_X}\left[ {\ln \left( {\frac{{1 + {H_{\rm{E}}}\varsigma _{\rm{E}}^2X}}{{1 + {H_{\rm{B}}}\varsigma _{\rm{B}}^2X}}} \right)} \right] \nonumber \\
 &&= \frac{1}{2}\ln \left( {\frac{{{H_{\rm{E}}}\varsigma _{\rm{E}}^2}}{{{H_{\rm{B}}}\varsigma _{\rm{B}}^2}}} \right).
\label{eq107}
\end{eqnarray}
Therefore, the asymptotic lower bound when $\alpha  \ne 0.5$, we have
\begin{eqnarray}
\mathop {\lim }\limits_{A \to \infty } {C'_{{\rm{Low}}}} \!\!\!\!\!&=&\!\!\!\!\! \mathop {\lim }\limits_{A \to \infty } \left\{ {\frac{1}{2}\ln \left( {\frac{{{H_{\rm{B}}}{H_{\rm{E}}}\varsigma _{\rm{E}}^2\sigma _{\rm{E}}^2}}{{\varsigma _{\rm{B}}^2\sigma _{\rm{B}}^2}}\frac{{{{({e^{cA}} - 1)}^2}}}{{{c^2}}}} \right) - c\xi P} \right. \nonumber \\
&&\!\!\!\!\!\left. { - \frac{1}{2}\ln \left[ {2\pi e\left( {H_{\rm{E}}^2\left( {\frac{{A\left( {cA - 2} \right)}}{{c\left( {1 - {e^{ - cA}}} \right)}} + \frac{2}{{{c^2}}} - {\xi ^2}{P^2}} \right) + {H_{\rm{E}}}\xi P\varsigma _{\rm{E}}^2\sigma _{\rm{E}}^2 + \sigma _{\rm{E}}^2} \right)} \right]} \right\} \nonumber \\
 &=&\!\!\!\!\! \mathop {\lim }\limits_{A \to \infty } \!\frac{1}{2}\!\ln \!\!\left[\! {\frac{{{H_{\rm{B}}}{H_{\rm{E}}}\varsigma _{\rm{E}}^2\sigma _{\rm{E}}^2{{({e^{cA}} - 1)}^2}}}{{2\pi e{c^2}\varsigma _{\rm{B}}^2\sigma _{\rm{B}}^2{e^{2c\xi P}}\!\left[ {H_{\rm{E}}^2\!\left( {\frac{{A\left( {cA - 2} \right)}}{{c\left( {1 - {e^{ - cA}}} \right)}} \!+\! \frac{2}{{{c^2}}} \!-\! {\xi ^2}{P^2}} \right) \!\!+\!\! {H_{\rm{E}}}\xi P\varsigma _{\rm{E}}^2\sigma _{\rm{E}}^2 \!+\!\! \sigma _{\rm{E}}^2} \right]}}} \!\!\right]\!.
\label{eq108}
\end{eqnarray}

For the asymptotic upper bound, we have
\begin{eqnarray}
\mathop {\lim }\limits_{A \to \infty } {C'_{{\rm{Upp}}}} &=& \mathop {\lim }\limits_{A \to \infty } \frac{1}{2}\ln \left[ {\frac{{{H_{\rm{E}}}\varsigma _{\rm{E}}^2\sigma _{\rm{E}}^2\left( {{H_{\rm{B}}}A + \varsigma _{\rm{B}}^2\sigma _{\rm{B}}^2} \right)}}{{\varsigma _{\rm{B}}^2\sigma _{\rm{B}}^2\left( {H_{\rm{E}}^2A + \frac{{H_{\rm{E}}^2}}{{{H_{\rm{B}}}}}\varsigma _{\rm{B}}^2\sigma _{\rm{B}}^2 + M} \right)}}} \right] \nonumber\\
 &=& \frac{1}{2}\ln \left( {\frac{{{H_{\rm{B}}}\varsigma _{\rm{E}}^2\sigma _{\rm{E}}^2}}{{{H_{\rm{E}}}\varsigma _{\rm{B}}^2\sigma _{\rm{B}}^2}}} \right).
\label{eq109}
\end{eqnarray}
According to (\ref{eq106}), (\ref{eq108}) and (\ref{eq109}), we prove \emph{Corollary \ref{cor6}}. \quad\quad\quad\quad\quad\quad\quad\quad\quad\quad\quad\quad\quad\quad\quad$\blacksquare$

\end{document}